\documentclass{article}

\usepackage{amsfonts,amsmath,amsthm,amssymb,enumerate,xcolor,siunitx,multirow,hhline,hyperref,breakurl,mathdots,graphicx}
\usepackage[margin=1in]{geometry}
\usepackage[makeroom]{cancel}
\usepackage[nottoc]{tocbibind}

\newtheorem{problem}{Problem}
\newtheorem{theorem}{Theorem}
\newtheorem{example}{Example}
\newtheorem{lemma}{Lemma}
\newtheorem{corollary}{Corollary}
\newtheorem{remark}{Remark}
\newtheorem{definition}{Definition}

\newcommand{\paren}[1]{\left(#1\right)}
\newcommand{\bracket}[1]{\left[#1\right]}
\renewcommand{\brace}[1]{\left\{#1\right\}}
\renewcommand{\ang}[1]{\left\langle#1\right\rangle}
\newcommand{\floor}[1]{\left\lfloor#1\right\rfloor}
\newcommand{\ceil}[1]{\left\lceil#1\right\rceil}
\newcommand{\abs}[1]{\left|#1\right|}
\newcommand{\R}{\mathbb{R}}
\newcommand{\C}{\mathbb{C}}
\newcommand{\F}{\mathbb{F}}
\newcommand{\Z}{\mathbb{Z}}
\newcommand{\ring}{{\reflectbox{$\mathcal{R}$}}}
\newcommand{\M}[1]{\begin{bmatrix}#1\end{bmatrix}}
\newcommand{\MA}[2]{\bracket{\begin{array}{#1}#2\end{array}}}

\renewcommand{\cases}[1]{\brace{\begin{array}{ll}#1\end{array}}}

\newcommand{\rk}{\operatorname{rk}}
\newcommand{\brk}[1]{\underline{\rk}_{#1}}
\newcommand{\borderk}[1]{\underline{\operatorname{rk}}\paren{#1}}

\newcommand{\rowspan}[1]{\operatorname{rowspan}\paren{#1}}

\renewcommand{\ker}[1]{\operatorname{ker}\paren{#1}}
\renewcommand{\dim}[1]{\operatorname{dim}\paren{#1}}
\newcommand{\rref}[1]{\operatorname{rref}\paren{#1}}
\newcommand{\GL}[2]{\operatorname{GL}\paren{#1,\ #2}}
\renewcommand{\~}[1]{\widetilde{#1}}
\renewcommand{\vec}[1]{\overrightarrow{#1}}
\renewcommand{\Vec}[1]{\operatorname{vec}\paren{#1}}
\renewcommand{\span}[1]{\operatorname{span}\paren{#1}}
\newcommand{\cpdeval}[1]{\bracket{\bracket{#1}}}
\newcommand{\unfold}[2]{#1_{\paren{#2}}}

\newcommand{\ch}[1]{\operatorname{char}\paren{#1}}
\newcommand{\maxrank}{\boldsymbol{R}}
\newcommand{\borderring}[2]{\underline{#1}_{#2}}

\usepackage[noend]{algpseudocode}
\usepackage{algorithm,algorithmicx}

\newcommand*\Let[2]{\State #1 $\gets$ #2}
\algrenewcommand\algorithmicrequire{\textbf{Precondition:}}
\algrenewcommand\algorithmicensure{\textbf{Postcondition:}}

\newcommand{\Yield}[1]{\State \textbf{yield} #1}

\newcommand{\mainResult}{O^*\paren{|\F|^{(R-n_0)(\sum_{d\ge 1} n_d) \ + \ \min\paren{R,\ \sum_{d\ge 2} n_d}}}}
\newcommand{\borderResult}{O^*\paren{|\F|^{H\sum_{1\le r\le R} \sum_d \min(r,n_d)}}}

\newcommand{\DEPARTMENT}{Department of Electrical Engineering and Computer Science}
\newcommand{\DEGREE}{Master of Engineering in Electrical Engineering and Computer Science}
\newcommand{\SUPERVISOR}{Virginia Williams}
\newcommand{\SUPERVISORTITLE}{Professor of Electrical Engineering and Computer Science, Thesis supervisor}
\newcommand{\SUBMITDATE}{May 9, 2025}

\title{New results in canonical polyadic decomposition over finite fields}
\author{by \\ \\ Jason Yang}
\date{}

\begin{document}

\begin{center}
{\huge New results in canonical polyadic decomposition over finite fields}

$\phantom{}$

by

$\phantom{}$

Jason Yang

$\phantom{}$

S.B. Computer Science and Engineering

Massachusetts Institute of Technology, 2025

$\phantom{}$

Submitted to the

\DEPARTMENT

in partial fulfillment of the requirements for the degree of

$\phantom{}$

\DEGREE

$\phantom{}$

at the

$\phantom{}$

Massachusetts Institute of Technology

$\phantom{}$

May 2025

$\phantom{}$

\textcopyright \ 2025 Jason Yang.
This work is licensed under a \href{https://creativecommons.org/licenses/by/4.0/}{CC BY 4.0} license.

$\phantom{}$

The author hereby grants to MIT a nonexclusive, worldwide, irrevocable, royalty-free license to exercise any and all rights under copyright, including to reproduce, preserve, distribute and publicly display copies of the thesis, or release the thesis under an open-access license.
\end{center}

$\phantom{}$

$\phantom{}$

$\begin{array}{ll}
    \textrm{Authored by:} & \textrm{Jason Yang} \\
     & \textrm{\DEPARTMENT} \\
     & \textrm{\SUBMITDATE} \\
     & \\
    \textrm{Certified by:} & \textrm{\SUPERVISOR} \\
     & \textrm{\SUPERVISORTITLE} \\
     & \\
    \textrm{Accepted by:} & \textrm{Katrina LaCurts} \\
      & \textrm{Chair, Master of Engineering Thesis Committee}
\end{array}$


\newpage

\maketitle

\begin{center}
Submitted to the
\DEPARTMENT $ $
on \SUBMITDATE $ $

in Partial Fulfillment of the Requirements for the Degree of

\DEGREE

$\phantom{}$
\end{center}

\subsection*{Abstract}
Canonical polyadic decomposition (CPD) consists of expressing a tensor (multidimensional array) as a sum of several rank-1 tensors, each of which is an outer/separable product of vectors.
The number of rank-1 tensors used in a CPD is called the rank of the CPD, and the minimum possible rank of a CPD for a given tensor is called the rank of the tensor.
CPD is at the core of fast matrix multiplication, a computational problem with widespread implications across several seemingly unrelated problems in computer science.
Much recent progress in this field has used randomized heuristic search to find new CPDs, often over a finite field.
However, if these techniques fail to find a CPD with low enough rank, they cannot prove that no such CPD exists.
Consequently, these methods fail to resolve certain long-standing questions, such as whether the tensor corresponding to $3\times 3$ matrix multiplication has rank less than 23.

To make progress on these problems, we develop a novel algorithm that preserves exactness, i.e. they can provably verify whether or not a given tensor has a specified rank.
Compared to brute force, when searching for a rank-$R$ CPD of a $n_0\times\dots\times n_{D-1}$-shaped tensor over a finite field $\F$, where $n_0\ge \dots\ge n_{D-1}$, our algorithm saves a multiplicative factor of roughly $|\F|^{R(n_0-1)+n_0(\sum_{d\ge 1} n_d)}$.
Additionally, our algorithm runs in polynomial time.
We also find a novel algorithm to search border CPDs, a variant of CPDs that is also important in fast matrix multiplication.

Finally, we study the maximum rank problem and give new upper and lower bounds, both for families of tensor shapes and specific shapes. 
Although our CPD search algorithms are still too slow to resolve the rank of $3\times 3$ matrix multiplication, we are able to utilize them in this problem by adding extra search pruners that do not affect exactness or increase asymptotic running time.

$\phantom{0}$

\noindent Thesis supervisor: \SUPERVISOR

\noindent Title: \SUPERVISORTITLE

\newpage

\section*{Acknowledgments}
I am deeply grateful to the following people, each of whom have made this thesis possible:

My supervisor Prof. Virginia Vassilevska Williams, as she guided me for the initial version of this research project --- studying fast matrix multiplication --- throughout all my undergraduate years, and during this thesis.
She acted as my main source of knowledge for this topic,
continually suggested conjectures, patterns, or other interesting things to investigate,
and helped me gain progress with this project via weekly meetings.

Prof. Erik Demaine, for which his class \href{https://courses.csail.mit.edu/6.5440/fall23/}{6.5440 Algorithmic Lower Bounds: Fun with Hardness Proofs}\footnote{https://courses.csail.mit.edu/6.5440/fall23/} motivated me to branch out beyond fast matrix multiplication and study tensor decomposition in a more general setting.

Austin Conner, for resolving an important lemma for the border variant of tensor decomposition.

Everyone, including anonymous colleagues, who reviewed my thesis and relevant papers, and suggested improvements to writing and organization.

My family, for supporting me throughout my MIT journey and keeping me in good health.

\newpage

\tableofcontents

\newpage

\listoftables

\newpage

\section{Introduction}
Given a tensor (multidimensional array) $T\in \ring^{n_0\times\dots\times n_{D-1}}$ over a ground ring $\ring$, a \textit{rank-$R$ canonical polyadic decomposition (CPD)} of $T$ is a list of matrices $A_d\in\ring^{n_d\times R},\ 0\le d<D$ such that \[T = \M{\sum_{0\le r<R} \prod_{0\le d<D} (A_d)_{i_d,r}}_{i_0,\dots,i_{D-1}},\]

i.e. the $(i_0,\dots,i_{D-1})$ entry of $T$ is $\sum_{0\le r<R} \prod_{0\le d<D} (A_d)_{i_d,r}$.

We call the $A_d$ ``factor matrices" and abbreviate the right-hand side as $\cpdeval{A_0,\dots,A_{D-1}}$. Note that this expression is also equal to $\sum_r \bigotimes_d (A_d)_{:,r}$, where $\otimes$ denotes the tensor product.

The notation we use for tensors is typical in computer science.
In algebraic geometry, tensors are typically notated as multilinear forms over formal variables.
Specifically, the above tensor $T$ would be the $D$-linear form $\sum_{i_0,\dots,i_{D-1}} T_{i_0,\dots,i_{D-1}}(\alpha_{0,i_0}\cdots \alpha_{D-1,i_{D-1}})$ for formal variables $\alpha_{d,i_d}$, and the CPD expression $\cpdeval{A_0,\dots,A_{D-1}}$ would be $\sum_r \paren{\sum_{i_0} (A_0)_{i_0,r}\alpha_{0,i_0}}\cdots\paren{\sum_{i_{D-1}} (A_{D-1})_{i_{D-1,r}}\alpha_{D-1,i_{D-1}}}$.

The rank of $T$, denoted $\rk_\ring\paren{T}$, is the smallest $R$ such that there exists a rank-$R$ CPD of $T$ over the ring $\ring$. Determining tensor rank is the central problem underlying fast matrix multiplication \cite{fmm-survey}. Formally, the action of multiplying a $m\times k$ matrix with a $k\times n$ matrix can be represented as a $mk\times kn\times nm$ tensor commonly denoted $\ang{m,k,n}$ \footnote{Typically, $\ang{m,k,n}$ represents the action of matrix multiplication followed by a transpose, which makes certain symmetries of the tensor easier to notate. This does not change the rank of this tensor, since applying a transpose to the output is equivalent to permuting the slices of $\ang{m,k,n}$ along axis 2 (the $nm$-length axis).}; then a rank-$R$ CPD of this tensor can be converted into a divide-and-conquer algorithm for multiplying two $N\times N$ matrices in $O(N^{3\log_{mkn} R})$ time. The quantity $3\log_{mkn} R$ is known as the \textit{running time exponent} of such a CPD. The famous Strassen algorithm \cite{strassen} corresponds to a rank-7 CPD of $\ang{2,2,2}$.
Furthermore, every fast matrix multiplication algorithm corresponds to a CPD of some $\ang{m,k,n}$, as long as the algorithm is restricted to arithmetic operations \cite{fmm-survey}.

The asymptotically fastest known algorithm for matrix multiplication \cite{fmm-record} and its predecessors correspond to CPDs of very large $\ang{m,k,n}$ tensors; this is a consequence of applying several algebraic techniques, each of which produces a \textit{sequence} of CPDs whose running time exponents converge to some limit.
As a result, the constant factors of such algorithms render them impractical \cite{fmm-survey}, despite new developments to mitigate this issue \cite{fmm-leading-constant}.
Furthermore, it is known that the optimal running time exponent must be a limit, i.e. it cannot be achieved by any single CPD \cite{cw-limit}.

An alternative approach to fast matrix multiplication is to directly find low-rank CPDs of small $\ang{m,k,n}$ tensors, which will yield suboptimal running time exponents but hopefully manageable constant factors. Much work has been done in this direction using computer search \cite{smirnov, courtois-333, heule-sat, alphatensor, flip, flip2, adaflip, deza-constraint-programming}, which we detail in Section \ref{sec:prior-work}.
We are most interested in the latter approach to fast matrix multiplication, due to impractical constant factors in the former approach. However, we deviate from previous research in two ways:
\begin{enumerate}[1.]
    \item We restrict ourselves to \textit{exact} algorithms for finding low-rank CPDs. Although there has been exciting progress with heuristic search methods in fast matrix multiplication \cite{courtois-333, smirnov, heule-sat, alphatensor, flip, symmflip}, these methods still have not resolved questions such as whether $\rk(\ang{3,3,3})<23$, which has been open since 1976 \cite{laderman-333}.

    \item We generalize to arbitrary tensors, not just those for matrix multiplication. Our reasons for doing so are: to find general search optimizations in tensor CPD that may have been overlooked; and to make our algorithms applicable to problems besides matrix multiplication that use CPD, such as the light bulb problem \cite{light-bulb} and dynamic programming \cite{fine-grained-dp}.
\end{enumerate}

To ensure an exact algorithm is possible in principle, we restrict the ground ring to be a finite field, so that the search space is finite. Doing so is not too strong of a handicap, as some previous work \cite{alphatensor, flip, flip2} has successfully lifted novel CPDs of matrix multiplication tensors from a finite field to the integers.

It is likely impossible to solve tensor rank in polynomial time, as tensor rank is NP-hard over the finite fields (for arbitrary tensors) \cite{tensor-np}.
However, we are still interested in improving the runtime as much as we can.
Our main result, developed over the course of \cite{yang24, yang25}, is:

\begin{theorem}
\label{thm:main}
Given a concise tensor $T\in\F^{n_0\times\dots\times n_{D-1}}$, finding a rank-$R$ CPD of $T$ or determining that none exists can be done in $\mainResult$ time and $O^*(1)$ space.
\end{theorem}

We explain our proof of Theorem \ref{thm:main} and the definition of ``concise" in Section \ref{sec:cpd}.
We also detail several search pruning strategies that significantly speed up our algorithm in practice.

Compared to brute force, which would run in $O^*\paren{|\F|^{R(\sum_d n_d)}}$ time, Theorem \ref{thm:main} decreases the exponent by $Rn_0+n_0(\sum_{d\ge 1} n_d)-\min(R,\sum_{d\ge 2} n_d)$; for large enough $D$ and $n_d$, this simplifies to $R(n_0-1)+n_0(\sum_{d\ge 1} n_d)$.

We also study \textit{border} CPDs, which are essentially CPDs over the polynomial ring $\F[x]/(x^H)$ for some positive integer $H$, and are equally important as ordinary CPDs in terms of minimizing the asymptotic complexity of fast matrix multiplication (more detail in Section \ref{sec:border}). Over this ring we achieve the following:

\begin{theorem}
\label{thm:border}
Given a concise tensor $T\in\paren{\F[x]/(x^H)}^{n_0\times\dots\times n_{D-1}}$, finding a rank-$R$ CPD of $T$ or determining that none exists can be done in $\borderResult$ time and $O^*(1)$ space.
\end{theorem}
Compared to brute force, which would take $O^*\paren{|\F|^{HR\sum_d n_d}}$ time, Theorem \ref{thm:border} decreases the exponent term by roughly $\frac{1}{2}H\sum_d n_d^2$.

Finally, in Section \ref{sec:max-rank} we improve existing bounds on the maximum possible rank of a tensor with a given shape, some with the help of our algorithm for Theorem \ref{thm:main}.
Source code is available at \href{https://github.com/coolcomputery/tensor-cpd-search}{this URL}\footnote{https://github.com/coolcomputery/tensor-cpd-search}.

\section{Previous work}
\label{sec:prior-work}
To our knowledge, there are no previously known algorithms that find CPDs of arbitrary tensors with guaranteed correctness, besides Gr\"obner bases, which solves general systems of polynomial equations.
For CPDs over $\R$, there are some practical numerical methods (e.g. Jennrich's algorithm) that work for ``nice" tensors but not all.

Much previous work in this field has focused on finding CPDs of matrix multiplication (MM) tensors. This work can be divided into two camps: the first uses algebraic methods to indirectly construct CPDs for (very large) MM tensors, in order to reduce the MM exponent $\omega$ by any means necessary; the second directly searches for CPDs of small MM tensors.

A good overview of the algebraic camp for MM is given in \cite{fmm-survey}. For the search camp, we give a list of such approaches that we previously wrote in \cite{yang25}:

\begin{itemize}
    \item By hand: Strassen's algorithm was extended by Hopcroft and Kerr \cite{hopcroft-kerr-p2n} to prove $\rk(\ang{m,2,n})\le \ceil{\frac{3mn+\max\brace{m,n}}{2}}$.
    Laderman \cite{laderman-333} proved $\rk(\ang{3,3,3})\le 23$, which today is still the best known upper bound for $\ang{3,3,3}$. Many other upper bounds of small $\ang{m,k,n}$ are listed on Sedoglavic's table \cite{sedoglavic-table}.
    
    However, the most successful manual approach has been Pan's trilinear aggregation \cite{pan78, pan82}, consisting of several upper bounds on $\rk(\ang{n,n,n})$. The lowest running time exponent obtained by this approach is $\approx 2.7734$, from the rank bound $\rk(\ang{44,44,44})\le 36\ 133$.

    \item Numerical optimization: Smirnov \cite{smirnov} converts tensor CPD into an optimization problem whose objective is to minimize the sum-of-squares error between the target tensor $T$ and a CPD $\cpdeval{A,B,C}$ of fixed rank. The primary optimization technique is \textit{alternating least squares}, where one cycles through each factor matrix $A,B,C$, and analytically minimizes the objective with respect to that matrix while fixing all other factor matrices. Together with clever regularization, along with post-processing (to convert a sufficiently accurate CPD into an exact CPD), Smirnov proved $\rk(\ang{3,3,6})\le 40$ over the rationals (running time exponent $\approx 2.7743$), nearly surpassing trilinear aggregation \cite{pan82}.
    
    \item Boolean SAT: several works, such as Courtois et. al. \cite{courtois-333} and Heule et. al. \cite{heule-sat}, have searched for CPDs over the finite field $\F_2$ by formulating the problem as boolean SAT, specifically to try finding a rank-22 CPD of $\ang{3,3,3}$. Although unsuccessful, thousands of new rank-23 CPDs were discovered (over $\F_2$) that are inequivalent to Laderman's CPD.

    \item Reinforcement learning: Fawzi et. al. \cite{alphatensor} (AlphaTensor) formulates tensor CPD as a single-player game where one starts with a given tensor, a move consists of subtracting away a rank-1 tensor, and the goal is to end with all-zeros in as few moves as possible.
    Using a neural network together with Monte-Carlo tree search, several new CPDs were found, including $\rk(\ang{3,4,5})\le 47$ over the integers (running time exponent $\approx 2.8211$) and $\rk(\ang{4,4,4})\le 47$ over $\F_2$ (running time exponent $\approx 2.7773$); the latter result could not be lifted to the integers.

    \item Flip graphs: Kauers and Moosbauer \cite{flip, flip2} begin with an arbitrary CPD $T=\sum_r a_r\otimes b_r\otimes c_r$, then (when possible) iteratively apply the following ``flip" transformation to a randomly chosen pair of summands with matching factors in the first axis,

    \[a\otimes b\otimes c + a\otimes b'\otimes c'\]
    \[\rightarrow a\otimes (b+b')\otimes c + a\otimes b'\otimes (-c+c'),\]

    or an equivalent transformation with the tensor axes and the order of the summands permuted.
    It is worth noting that this is essentially an elementary row operation in a matrix factorization.
    
    After each iteration, a CPD is also ``reduced" when possible (detected using some conditions involving linear span) so that its rank decreases by 1.
    

    Kauers and Moosbauer first proved $\rk_{\F_2}(\ang{5,5,5})\le 95$ by initializing their search at the rank-96 CPD discovered by AlphaTensor \cite{alphatensor}.
    In a subsequent paper \cite{flip2}, the authors proved $\rk(\ang{2,6,6})\le 56$ over the integers by lifting a CPD from $\F_2$, marking the first improvement on $\ang{2,6,6}$ since Hopcroft and Kerr \cite{hopcroft-kerr-p2n}.

    Further modifications of this technique, such as adaptive flip graphs \cite{adaflip} and flip graphs with symmetry \cite{symmflip} have given stronger bounds.

    \item Constraint programming: closest to the spirit of our work, Deza et. al. \cite{deza-constraint-programming} use the IBM ILOG CP Optimizer to exhaustively search for CPDs with elements in $\brace{-1,0,1}$, with symmetry-breaking constraints added that do not affect correctness (e.g., forcing lexicographic order among the summands $a_r\otimes b_r\otimes c_r$). Although the authors do not find new upper bounds on tensor rank, they do recover the previously known bounds $\rk(\ang{2,2,2})\le 7$ and $\rk(\ang{2,2,3})\le 11$.
\end{itemize}

A table of upper bounds on $\rk(\ang{m,n,k})$ for many (thousands of) small $\ang{m,n,k}$ is given in \cite{sedoglavic-table}.
Despite this progress, the lowest known MM exponent that comes from a small tensor is still Pan's latest trilinear aggregation result \cite{pan82} from 1982, and its small improvement over Strassen's algorithm is likely outweighed by increased constant factors.

\subsection{Notation}
\label{sec:notation}
\begin{itemize}
    \item The statement $T:=\M{f(i_0,\dots,i_{D-1})}_{i_0,\dots,i_{D-1}}$ means that the $(i_0,\dots,i_{D-1})$ entry of tensor $T$ is $f(i_0,\dots,i_{D-1})$.

    \item Tensor slices are denoted with NumPy notation, and indices are 0-indexed.
    \begin{itemize}
        \item e.g., for a two-dimensional tensor $A$: $A_{0,:}$ is the first (topmost) row of $A$; and $A_{:,:c}$ denotes the submatrix containing all rows and the first (leftmost) $c$ columns.
    \end{itemize}

    \item The tensor product of $A\in\ring^{n_0\times\dots\times n_{D-1}}$ and $B\in\ring^{s_0\times\dots\times s_{E-1}}$
    is
    \[A\otimes B:=\M{A_{i_0,\dots,i_{D-1}} B_{j_0,\dots,j_{E-1}}}_{i_0,\dots,i_{D-1}, j_0,\dots,j_{E-1}}.\]

    The chain tensor product $A_0\otimes\dots\otimes A_{n-1}$ is denoted as $\bigotimes_{0\le i<n} A_i$.
    
    \item The axis-$d$ contraction of $T\in\ring^{n_0\times\dots\times n_{D-1}}$ by $M\in\ring^{n'_d \times n_d}$ is
    \[M\times_d T:=\M{\sum_{i_d} M_{i'_d,i_d} T_{_{i_0,\dots,i_{D-1}}}}_{i_0,\dots,i_{d-1}, i'_d, i_{d+1},\dots,i_{D-1}}.\]

    \item A rank-$R$ CPD is a list of factor matrices $A_d\in\ring^{n_d\times R},\ 0\le d<D$ that evaluates to the tensor $\cpdeval{A_0,\dots,A_{D-1}}:=\sum_r \bigotimes_d (A_d)_{:,r}$.

    \item $\rref{M}$ denotes the reduced row echelon form of matrix $M$.

    \item $\GL{n}{\F}$ denotes the set of invertible $n\times n$ matrices with elements in the field $\F$.

\end{itemize}

\section{Algorithms}
Here we describe our algorithms for both standard CPD and border CPD, which were cited in our previous works \cite{yang25} and \cite{yang24}, respectively. The key insight we use is to extract information about the input tensor $T$ via axis-operations:

\[M\times_d T:=\M{\sum_{i_d} M_{i'_d,i_d} T_{_{i_0,\dots,i_{D-1}}}}_{i_0,\dots,i_{d-1}, i'_d, i_{d+1},\dots,i_{D-1}},\]

defined for a suitably sized matrix $M$ (see notation in Section \ref{sec:notation}).

The following relation between axis-operations and CPDs is the crux of our algorithm and is easy to check by hand:


\[M\times_d \cpdeval{A_0,\dots,A_{D-1}}=\cpdeval{A_0,\dots,A_{d-1},MA_d,A_{d+1},\dots,A_{D-1}}.\]

A consequence is that applying an axis contraction on a tensor cannot increase its rank.

\subsection{CPD DFS}
\label{sec:cpd}
When $T$ has elements in a field $\F$, we can reduce $T$ in a manner similar to row reduction for matrices. More precisely, for each axis $d$:
\begin{itemize}
    \item unfold $T$ into the matrix $\unfold{T}{d}:=\MA{c}{\vdots\\\hline \Vec{T_{\cdots,:,i_d,:,\cdots}}\\\hline \vdots}_{i_d}$;
    \item find some invertible matrix $Q$ s.t. $Q \unfold{T}{d} = \rref{\unfold{T}{d}}$;
    \item and replace $T$ with $Q\times_d T$.
\end{itemize}

Afterwards, remove any all-zeros slices from $T$; doing so does not change $\rk(T)$.

Since each such $Q$ can be constructed in polynomial time using Gaussian elimination,
and any rank-$R$ CPD of the original tensor $T$ can be converted into a CPD of the new $T$ (and vice versa) by applying axis-operations,
we can assume without loss of generality (WLOG) that the input tensor $T$ with shape $n_0\times\dots\times n_{D-1}$ satisfies $\rk(T_{(d)})=n_d$; such a tensor is called \textit{concise}:
\begin{definition}
A tensor $T\in\F^{n_0\times\dots\times n_{D-1}}$ is concise if for each axis $d$, the unfolding $\unfold{T}{d}:=\MA{c}{\vdots\\\hline \Vec{T_{\cdots,:,i_d,:,\cdots}}\\\hline \vdots}_{i_d}$ satisfies $\rk(\unfold{T}{d})=n_d$, i.e. $\unfold{T}{d}$ has full row rank.
\end{definition}

It is also clear that if a concise tensor $T$ with rank $\le R$ must have each side length be $\le R$:
\begin{lemma}
If $T\in\F^{n_0\times\dots\times n_{D-1}}$ is concise and $\rk(T)\le R$, then $\forall d: n_d\le R$.
\end{lemma}
\begin{proof}
Suppose $T=\cpdeval{A_0,\dots,A_{D-1}}$ is a rank-$R$ CPD. Then unfolding along axis $d$ yields the relation $T_{(d)}=A_d \MA{c}{\vdots \\\hline \Vec{\otimes_{d'\ne d} (A_{d'})_{:,r}} \\\hline \vdots}_{r}$, implying $\rk(T_{(d)})\le R$.
\end{proof}

An immediate consequence is that finding a rank-1 CPD is easy:
\begin{lemma}
\label{rank1}
Given an arbitrary tensor $T\in\F^{n_0\times\dots\times n_{D-1}}$, returning a rank-1 CPD of $T$ or determining that no such CPD exists can be done in $O^*(1)$ time.
\end{lemma}
\begin{proof}
WLOG assume $T$ is concise; then $\rk(T)\le 1$ if and only if $\forall d:\ n_d\ne 1$.
\end{proof}

Now we describe the main algorithm of Theorem \ref{thm:main}.
Suppose a concise tensor $T$ has a rank-$R$ CPD $\cpdeval{A_0,\dots,A_{D-1}}$. By performing row reduction on $A_0$, there exists some $Q\in\GL{n_0}{\F}$ such that $QA_0=\rref{A_0}$. Since $T$ is concise, $\rk(A_0)=n_0$, so $\rref{A_0}=\MA{c|c}{I_{n_0} & C}$ for some $C$, up to simultaneous permutation of columns of all $A_d$.

Consider isolating the $I_{n_0}$ portion as follows:
\[Q\times_0 T
=\cpdeval{I_{n_0}, (A_1)_{:,:n_0},\dots,(A_{D-1})_{:,:n_0}}+\cpdeval{C, (A_1)_{:,n_0:},\dots,(A_{D-1})_{:,n_0:}}\]
\[\Leftrightarrow Q\times_0 T-\cpdeval{C, (A_1)_{:,n_0:},\dots,(A_{D-1})_{:,n_0:}}
=\cpdeval{I_{n_0}, (A_1)_{:,:n_0},\dots,(A_{D-1})_{:,:n_0}}\]
\[\Leftrightarrow \~{Q}\times_0 \~{T}=\cpdeval{I_{n_0}, (A_1)_{:,:n_0},\dots,(A_{D-1})_{:,:n_0}},\]

where $\~{Q}=\M{Q|C}$, $\~{T}_{:n_0,:,\cdots}=T$, and $\~{T}_{r,:,\cdots}=-\bigotimes_{d\ge 1} (A_d)_{:,r}$ for $n_0\le r<R$.

The right-hand side of the last equation is a tensor with shape $n_0\times\dots\times n_{D-1}$ whose $i$-th slice along axis 0 is $\otimes_{d\ge 1} (A_d)_{:,i}$, since the $i$-th row of $I_{n_0}$ has a one at column $i$ and zeros everywhere else.
Because each of these slices is a rank-1 tensor and uses a disjoint subset of unknowns, the last equation is equivalent to

\[\forall 0\le i<n_0:\ \rk((\~{Q}\times_0 \~{T})_{i,:,\cdots})\le 1\]

\[\Leftrightarrow \forall 0\le i<n_0:\ \rk(\~{Q}_{i,:}\times_0 \~{T})\le 1.\]

Suppose we knew $\~{T}$ and just had to solve for $\~{Q}$; then the following conditions on $\~{Q}$ are necessary and sufficient:
\begin{itemize}
    \item each row of $\~{Q}$ is in the set $S:=\brace{v\in\F^{1\times R} : \rk(v\times_0 \~{T})\le 1}$;
    \item $\~{Q}_{:,:n_0}$ is invertible;
\end{itemize}

There are two ways to solve for $\~{Q}$:
\begin{itemize}
    \item Construct $S$ directly (by iterating over all $v$) and check if the matrix $M:=\MA{c}{\vdots\\\hline v\\\hline \vdots}_{v\in S}$ satisfies $\rk(M_{:,:n_0})=n_0$.
    
    If so, we can use Gaussian elimination to extract $n_0$ many distinct row indices $i_0,\dots,i_{n_0-1}$ in $M$ such that $\MA{c}{M_{i_0,:n_0}\\\hline \vdots\\\hline M_{i_{n_0-1},:n_0}}$ is invertible,
    and then set $\~{Q}=\MA{c}{M_{i_0,:}\\\hline \vdots\\\hline M_{i_{n_0-1},:}}$.

    By Lemma \ref{rank1}, we can check for an arbitrary $v$ whether $\rk(v\times_0 \~{T})\le 1$ in polynomial time. Thus, constructing $S$ takes $O^*(|\F|^R)$ time. Finally, constructing $M$ and $i_0,\dots,i_{n_0-1}$ increases total running time by a polynomial multiplicative factor.

    \item Replace the condition $\rk(v\times_0 \~{T})\le 1$ with $v\times_0 \~{T}=\bigotimes_{d\ge 1} u_d$, where $u_d\in\F^{n_d},\ d\ge 1$ are new vector variables.

    If we fix $u_2,\dots,u_{D-1}$, the equation becomes linear with respect to $(v,u_1)$, so it is possible to construct a basis of these solutions $B(u_2,\dots,u_{D-1})\subseteq \F^R\times \F^{n_0}$ in polynomial time.
    To take advantage of this property, consider iterating over all tuples $(u_d)_{d\ge 2}$ and constructing the set $\~{S}:=\brace{v:\ (v,u_1)\in B(u_2,\dots,u_{D-1}),\ u_2\in\F^{n_2}, \dots,u_{D-1}\in\F^{D-1}}$; then $\span{\~{S}}=\span{S}$, so we can use $\~{S}$ in place of $S$ when constructing $M$.

    Because each basis $B(u_2,\dots,u_{D-1})$ contains $\le R$ many elements, constructing $\~{S}$ takes $O^*\paren{|\F|^{\sum_{d\ge 2} n_d}}$ time.
\end{itemize}

Once we have $\~{Q}$, we can determine $(A_d)_{:,:n_0}$ by solving
$(\~{Q}\times_0 \~{T})_{i,:,\cdots}=\otimes_{d\ge 1} (A_d)_{:,i}$ for $0\le i<n_0$ (finding several independent rank-1 CPDs)
and setting $A_0=Q^{-1}\MA{c|c}{I_{n_0} & C}$, where $(Q,C)=(\~{Q}_{:,:n_0}, \~{Q}_{:,n_0:})$.

Finally, since we had to fix $\~{T}$ during this procedure, we just repeat with every possible $\~{T}$, which is equivalent to fixing all possible assignments of $(A_d)_{:,n_0:},\ d\ge 1$.

To analyze total runtime,
we fix $\~{T}$ to $|\F|^{(R-n_0)\sum_{d\ge 1} n_d}$ many assignments,
then solve for $\~{Q}$ using the faster of the two methods \footnote{Deciding which method of solving $\~{Q}$ is asymptotically faster can be done by comparing the quantities $R$ and $\sum_{d\ge 2} n_d$.} in $O^*\paren{|\F|^{\min\paren{R,\ \sum_{d\ge 2} n_d}}}$ time. Thus, the total runtime is
\[\mainResult.\]

With a small modification, the space complexity is $O^*(1)$, as we can generate each $v$ in $S$ (or $\~{S}$) on the fly and append $v$ to $\~{Q}$ only if doing so increases $\rk(\~{Q}_{:,:n_0})$.

To minimize the time complexity, we can permute the axes of $T$ beforehand so that $n_0\ge\dots\ge n_{D-1}$.
In addition, we can speed up by a multiplicative factor of $(R-n_0)!$, since $\~{T}$ is invariant under simultaneous permutation of columns of $(A_d)_{:,n_0:},\ d\ge 1$;
for example, we can force the vector tuples $\paren{(A_d)_{:,r}}_{d\ge 1}$ to be lexicographically sorted w.r.t. $n_0\le r<R$.

Algorithm \ref{alg:main} shows pseudocode for implementing this procedure; the matrix $\~{Q}=\MA{c|c}{Q&C}$ is represented as the pair $(Q,C)$, and matrix slices $(A_d)_{:,:n_0}, (A_d)_{:,n_0:}$ are represented as $X_d, Y_d$, respectively.

In \cite{yang25}, we describe a variant of Algorithm \ref{alg:main} that improves the exponent in the runtime shown above by roughly an $O(R)$ additive amount. However, this variant is incompatible with the aforementioned lex-sorting optimization that saves a multiplicative factor of $(R-n_0)!$, making it slower overall, so we omit it here.

\begin{algorithm}
    \caption{Search algorithm for Theorem \ref{thm:main}}
    \label{alg:main}
    \begin{algorithmic}[1]
        \Require{
            $T\in\F^{n_0\times\dots\times n_{D-1}}$, $R\ge 0$;
            $T$ concise, $n_0\ge \dots\ge n_{D-1}$
        }
        \Ensure{
            returns a rank-$\le R$ CPD of $T$ as a sequence of factor matrices, if one exists;
            else, returns null
        }
        \Function{search}{$T, R$}
            \If{$R<n_0$}
                \State \Return{\textbf{null}}
            \EndIf
            \For{$1\le d<D,\ Y_d\in\F^{n_d\times (R-n_0)}$}
                \Let{$\Delta$}{\Call{test}{$T,R,(Y_d)_{1\le d<D}$}}
                \If{$\Delta$ \textrm{not null}}
                    \State \Return{$\Delta$}
                \EndIf
            \EndFor
            \State \Return{\textbf{null}}
        \EndFunction
        
        \Function{test}{$T,R,(Y_d)_{1\le d<D}$}
            \Let{$(Q,C)$}{$([],[])$}
            \For{$(v,c)\in$ \Call{good\_pairs}{$T,R,(Y_d)_{1\le d<D}$}}
                \If{$\rk\paren{\MA{c}{Q\\\hline v}}>\rk\paren{Q}$}
                    \Let{$(Q,C)$}{$\paren{\MA{c}{Q\\\hline v}, \MA{c}{C\\\hline c}}$}
                \EndIf
            \EndFor
            \If{$Q$ has $n_0$ many rows}
                \Let{$(X_d)_{1\le d<D}$}{$([])_{1\le d<D}$}
                \For{$0\le i<n_0$}
                    \Let{$T_\square$}{$Q_{i,:}\times_0 T - \cpdeval{C_{i,:},Y_1,\dots,Y_{D-1}}$}
                    \Let{$(u_d)_{1\le d<D}$}{\Call{rank1}{$T_\square$}} \Comment{obtainable with Lemma \ref{rank1}}
                    \Let{$(X_d)_{1\le d<D}$}{$\paren{\MA{c|c}{X_d&u_d}}_{1\le d<D}$}
                \EndFor
                \State \textbf{return} $\Big(Q^{-1}\MA{c|c}{I_{n_0}&C}, \MA{c|c}{X_1&Y_1},\dots,\MA{c|c}{X_{D-1}&Y_{D-1}}\Big)$
            \EndIf
            \State \Return{\textbf{null}}
        \EndFunction
        \Function{good\_pairs}{$T,R,(Y_d)_{1\le d<D}$}
            \If{$R \le \sum_{d\ge 2} n_d$}
                \For{$v\in \F^{1\times n_0},\ c\in \F^{1\times (R-n_0)}$}
                    \If{$\rk\paren{v\times_0 T - \cpdeval{c,Y_1,\dots,Y_{D-1}}}\le 1$}
                        \Yield{$(v,c)$} \Comment{emit output without storing in memory}
                    \EndIf
                \EndFor
            \Else
                \For{$2\le d<D,\ u_d\in\F^{n_d}$}
                    \State $f\gets\Big(\MA{c|c|c}{v&c&u_1}\mapsto v\times_0 T - \cpdeval{c,Y_1,\dots,Y_{D-1}}-u_1\otimes\paren{\bigotimes_{2\le d<D} u_d}\Big)$
                    \Let{$M$}{matrix s.t. $\rowspan{M}=\ker{f}$} \Comment{obtainable with Gaussian elimination}
                    \For{$0\le i<(\# \textrm{ rows in } M)$}
                        \Yield{$(M_{i,:n_0},M_{i,n_0:R})$}
                    \EndFor
                \EndFor
            \EndIf
        \EndFunction
    \end{algorithmic}
\end{algorithm}

\subsubsection{Pruners}
\label{sec:pruners}
Algorithm \ref{alg:main} can be implemented in a recursive depth-first search fashion, by fixing each tuple of columns $((A_d)_{:,r})_{d\ge 1}$ for each $r$ at a time, e.g., $r=n_0,\dots,R-1$.
By examining more closely what we fix these variables to, we can prune some subtrees of the search space and achieve significant practical speedups.

Suppose we have currently fixed $((A_d)_{:,n_0:r})_{d\ge 1}$ for some $n_0\le r\le R$: we can take the condition $\paren{\forall i: \rk\paren{(\~{Q}_{i,:}\times_0 \~{T})\le 1}}$ that our algorithm would check if all $((A_d)_{:,n_0:})_{d\ge 1}$ were fixed, isolate the fixed variables from the indeterminate variables as much as possible, then weaken the condition so that only the fixed variables remain.
The remaining expression is what we call a ``pruning condition" or ``pruner".
We will call a pruner ``practical" if its running time is $O^*(|\F|^{O(R+\sum_d n_d)})$, i.e. the exponent is at most linear in the rank threshold $R$ and shape lengths $n_d$ of the tensor $T$.

We avoid applying pruners in base cases $r=R$, since Algorithm \ref{alg:main} already resolves each of those search states in $O^*(|\F|^R)$ time.

In this section, we list three pruners, all of which we use in our Java implementation of Algorithm \ref{alg:main}. All three pruners are practical if $D=3$, which is the case we are most interested in.

\begin{lemma}[``Rref-pruning"]
\label{rref-prune}
Let $T\in\F^{n_0\times\dots\times n_{D-1}}$ be a concise tensor, and suppose we have fixed thte submatrices $((A_d)_{:,n_0:r})_{d\ge 1}$ for some $n_0\le r\le R$. If there exists a rank-$R$ CPD $T=\cpdeval{A_0,\dots,A_{D-1}}$ that is consistent with what we have fixed, then the set 
\[S':=\brace{v\in\F^{1\times r}:\ \rk(v\times_0 T')\le R-r+1},\]
where $T'$ is the augmented tensor such that $T'_{:n_0,:,\cdots}=T$ and $T'_{r',:,\cdots}=-\bigotimes_{d\ge 1} (A_d)_{:,r'}$ for $n_0\le r'<r$,
contains $n_0$ many vectors $v_0,\dots,v_{n_0-1}$ such that the matrix $\MA{c}{\vdots\\\hline (v_i)_{:,:n_0}\\\hline \vdots}_i$ has rank $n_0$.
\end{lemma}
\begin{proof}
Let $\~{T}$ be the augmented tensor such that $\~{T}_{:n_0,:,\cdots}=T$ and $\~{T}_{r',:,\cdots}=-\bigotimes_{d\ge 1} (A_d)_{:,r'}$ for $n_0\le r'<R$.

Using similar reasoning as in Section \ref{sec:cpd}, a CPD exists if and only $S:=\brace{v\in\F^{1\times R}:\ \rk(v\times_0 \~{T})\le 1}$ contains $n_0$ many vectors $v_0,\dots,v_{n_0-1}$ such that the matrix $\MA{c}{\vdots\\\hline (v_i)_{:,:n_0}\\\hline \vdots}_i$ has rank $n_0$.

We want to separate the fixed portion $(A_d)_{:,n_0:r}$ of each factor matrix from the unfixed portion $(A_d)_{:,r:}$. To do this, we rewrite $v\times_0 \~{T}$ as $v_{:,:r}\times_0 \~{T}_{:r,:,\cdots} + v_{:,r:}\times_0 \~{T}_{r:,:,\cdots}$,
then note that 
having $\rk(v\times_0 \~{T})\le 1$ implies $\rk(v_{:,:r}\times_0 \~{T}_{:r,:,\cdots})\le R-r+1$.
Finally, since $T'=\~{T}_{:r,:,\cdots}$, we have $S\subseteq S'$.

Since $r\ge n_0$, any satisfying subset of vectors $\brace{v_i}_i\subseteq S$ would yield a satisfying subset $\brace{(v_i)_{:,:r}}_i \in S'$, since $\paren{(v_i)_{:,:r}}_{:,:n_0}=(v_i)_{:,:n_0}$.
\end{proof}

We call Lemma \ref{rref-prune} ``rref-pruning", since it indirectly relies on the fact that the factor matrix $A_0$ has a reduced row-echelon form. When $D=3$, rref-pruning can be implemented to run in $O^*(|\F|^r)$ time and $O^*(1)$ space, since constructing $S'$ consists of evaluating the ranks of many matrices, and we can generate elements of $S'$ on the fly in a similar manner as in Algorithm \ref{alg:main}.

The next pruner relies on the following result previously discovered by \cite{laskowski}:
\begin{theorem}[\cite{laskowski}]
\label{laskowski}
If a concise tensor $T\in\F^{n_0\times\dots\times n_{D-1}}$ has rank $\le R$, then $s:=\sum_{v\in\F^{1\times n_0}} \rk(v\times_0 T)$ is at most $R(|\F|^{n_0}-|\F|^{n_0-1})$.
\end{theorem}
\begin{proof}
Let $\cpdeval{A_0,\dots,A_{D-1}}$ be a rank-$R$ CPD of $T$.
Then $v\times_0 T=\cpdeval{vA_0,A_1,\dots,A_{D-1}}$,
so $\rk(v\times_0 T)$ is at most the number of nonzeros in $vA_0$, which we denote as $\#_{\ne 0}(vA_0)$.
We can upper-bound the sum $s$ with

\[
s\le \sum_{v\in\F^{1\times n_0}} \#_{\ne 0}(vA_0)
\]
\[
=\sum_{v\in\F^{1\times n_0},\ 0\le r<R} 1[v(A_0)_{:,r}\ne 0]
\]
\[
=\sum_{0\le r<R} |\F|^{n_0} - \abs{\textrm{left\_nullspace}((A_0)_{:,r})};
\]

since $\rk\paren{(A_0)_{:,r}}\le 1$, the rank-nullity theorem implies $\dim{\textrm{left\_nullspace}((A_0)_{:,r})}\ge n_0-1$, so
\[
s\le \sum_{0\le r<R} |\F|^{n_0} - |\F|^{n_0-1}
=R(|\F|^{n_0} - |\F|^{n_0-1}).
\]
\end{proof}

\begin{lemma}[``Lask-pruning"]
\label{lask-prune}
Let $T\in\F^{n_0\times\dots\times n_{D-1}}$ be a concise tensor, and suppose we have fixed $((A_d)_{i,:})_{d\ge 1}$. If there exists a rank-$R$ CPD $T=\cpdeval{A_0,\dots,A_{D-1}}$ that is consistent with what we have fixed, then
\[
\sum_{v\in\F^{1\times n_0}} \min_{w\in\F^{1\times (r-n_0)}} \rk\paren{v\times_0 T \ - \ \cpdeval{w, (A_1)_{:,n_0:r},\dots}}
\le (R-(r-n_0))(|\F|^{n_0}-|\F|^{n_0-1}).
\]
\end{lemma}
\begin{proof}
There must exist some assignment of $(A_0)_{:,n_0:r}$ such that $\rk\paren{T-\cpdeval{\dots,(A_d)_{:,n_0:r},\dots_d}}\le R-(r-n_0)$. Applying Theorem \ref{laskowski} directly yields

\[
s:=\min_{(A_0)_{:,n_0:r}\in\F^{n_0\times (r-n_0)}} \sum_{v\in\F^{1\times n_0}} \rk\paren{v\times_0 \paren{T-\cpdeval{\dots,(A_d)_{:,n_0:r},\dots_d}}}\le (R-(r-n_0))(|\F|^{n_0}-|\F|^{n_0-1}).
\]

Iterating over all possible $(A_0)_{:,n_0:r}$ would be too slow, so we look for a lower bound of the left-hand side that is easier to compute.
Expanding $v\times_0 \paren{T-\cpdeval{\dots,(A_d)_{:,n_0:r},\dots_d}}$ yields

\[
v\times_0 \paren{T-\cpdeval{\dots,(A_d)_{:,n_0:r},\dots_d}}
=v\times_0 T \ - \ \cpdeval{v(A_0)_{:,n_0:r}, (A_1)_{:,n_0:r},\dots};
\]

notice that $v(A_0)_{:,n_0:r}$ is a $1\times (r-n_0)$ vector, so it does not contain too many possible values. Consider replacing this with an arbitrary $1\times (r-n_0)$ vector that has complete freedom relative to $v$ and $(A_0)_{:,n_0:r}$; then
\[
s\ge \sum_{v\in\F^{1\times n_0}} \min_{w\in\F^{1\times (r-n_0)}} \rk\paren{v\times_0 T \ - \ \cpdeval{w, (A_1)_{:,n_0:r},\dots}}.
\]
\end{proof}

We call Lemma \ref{lask-prune} ``Laskowski-pruning" or ``lask-pruning" for short, named after the author of \cite{laskowski}. Similarly to rref-pruning, lask-pruning can be done in $O^*(|\F|^r)$ time and $O^*(1)$ space in the case $D=3$.

Both of the aforementioned pruners are ``negative", meaning they might prove whether a given search state is infeasible. In contrast, our third pruner is ``positive", meaning it might find a solution to a search state much faster than Algorithm \ref{alg:main} would by itself. We call this pruner ``rref-heuristic":
\begin{itemize}
    \item Suppose $T\in\F^{n_0\times\dots\times n_{D-1}}$ is the concise target tensor and we have fixed $((A_d)_{:,n_0:r})_{d\ge 1}$ for some $n_0\le r<R$;
    \item Let $T'$ be the augmented tensor such that $T'_{:n_0,:,\cdots}=T$ and $T'_{r',:,\cdots}=-\bigotimes_{d\ge 1} (A_d)_{:,r'}$ for $n_0\le r'<r$;
    \item Construct a list $L$ of all row vectors $v\in \F^{1\times r}$, sorted in nondecreasing order by $\rk(v\times_0 T')$;
    \item Initialize $\~{Q}$ to the $0\times n_0$ matrix;
    \item For each $v$ in $L$:
    add $v$ as a new row to $\~{Q}$ if doing so increases $\rk(\~{Q}_{:,:n_0})$;
    
    Afterwards, $\~{Q}$ is guaranteed to have $n_0$ many rows.
    
    \item For each $i$, compute a minimum-rank CPD of $\~{Q}_{i,:}\times_0 T'=\cpdeval{B^{(i)}_1,\dots,B^{(i)}_{D-1}}$;

    doing so allows us to express $\~{Q}\times_0 T'$ as
    $\sum_i e_i \otimes \cpdeval{B^{(i)}_1,\dots,B^{(i)}_{D-1}}$,
    where $e_i$ is 1-hot at index $i$ and $r_i:=\rk(\~{Q}_{i,:}\times_0 T')$
    
    \item Construct the CPD \[T
    =\paren{\sum_i (\~{Q}_{:,:n_0})^{-1}_{:,i}\otimes \cpdeval{B^{(i)}_1,\dots,B^{(i)}_{D-1}}} \
    + \ 
    \cpdeval{(\~{Q}_{:,:n_0})^{-1}\~{Q}_{:,n_0:},\ (A_1)_{:,n_0:r},\dots,(A_{D-1})_{:,n_0:r}};\]
\end{itemize}

The resulting CPD has rank $(r-n_0) + \sum_i \rk(\~{Q}_{i,:}\times_0 T)$.
If $D=3$, finding a min-rank CPD of $\~{Q}_{i,:}\times_0 T'$ would be equivalent to matrix rank factorization, which can be computed in polynomial time using Gaussian elimination, so rref-heuristic would run in $O^*(|\F|^r)$ time.

We show some examples of rref-pruning, lask-pruning, and rref-heuristic resolving CPD search states. For simplicity, we only include states where $r=n_0$, i.e. where we have not fixed any variables.

\begin{example}
Both rref-pruning and lask-pruning rule out the statement $\rk(T)\le 2$ for
$T=\M{\M{&1\\1}&\M{1\\&\phantom{0}}}$,
showing that $\rk(T)\ge 3$ over any finite field.
Rref-heuristic matches this lower bound.
\end{example}

\begin{example}
\label{addition-mod-2-tensor}
Both rref-pruning and lask-pruning prove that the tensor $T=\M{\M{&1\\1}&\M{1\\&1}}$ has rank $\ge \cases{3 & \operatorname{char}\ \F = 2 \\ 2 & \textrm{else}}$.
Rref-heuristic matches this lower bound.
\end{example}

\begin{example}
For any finite field of size $\ge n$, rref-heuristic can find a rank-$(2n-1)$ CPD of
the $(2n-1)\times n\times n$ tensor \[T=\M{
\M{1\\\phantom{0}&&\phantom{0}\\\phantom{0}&&\phantom{0}}
& \M{&1\\1\\\phantom{0}&&\phantom{0}}
& \dots
& \underbrace{\M{&&1\\&\iddots\\1}}_n
& \dots
& \M{\phantom{0}&&\phantom{0}\\&&1\\&1}
& \M{\phantom{0}&&\phantom{0}\\\phantom{0}&&\phantom{0}\\&&1}
}\]
corresponding to multiplication of $(n-1)$-th degree polynomials.
This is because for any $x\in\F$, $\M{1&\dots&x^{2n-2}}\times_0 T=\M{x^{i+j}}_{i,j}$ has rank 1.

On the other hand, $\rk(T)\ge 2n-1$, since the axis-0 slices of $T$ are linearly independent.
\end{example}

\begin{example}
Rref-pruning proves that the $n\times n\times n$ tensor
\[T=\M{
\M{1\\\phantom{0}&&\phantom{0}\\\phantom{0}&&\phantom{0}}
& \M{&1\\1\\\phantom{0}&&\phantom{0}}
& \dots
& \underbrace{\M{&&1\\&\iddots\\1}}_n
}\]
has rank $\ge 2n-1$ over any field.
To see why, set the rank threshold to $R=2n-2$; then we need $n$ many linearly independent row vectors $v$ such that $\rk(v\times_0 T)\le R-n+1=n-1$. But this forces $v_{n-1}=0$, which can be seen by considering matrix determinants, so $v$ is confined to a $(n-1)$-dimensional subspace.
\end{example}

\begin{remark}
Neither rref-pruning or lask-pruning is strictly stronger than the other. The below examples are over the field $\F_2$:
\begin{itemize}
    \item For $T_1=\M{\M{1\\&\phantom{0}\\&&\phantom{0}}&\M{1&1\\&\phantom{0}\\&&\phantom{0}}&\M{&&1\\&1\\1}}$, rref-pruning proves $\rk(T_1)\ge 5$ whereas lask-pruning proves $\rk(T_1)\ge 4$.
    \item For $T_2=\M{\M{1\\&\phantom{0}\\&&\phantom{0}}&\M{1\\&1\\&&1}&\M{&&1\\&&1\\1&1}}$, rref-pruning proves $\rk(T_2)\ge 4$ whereas lask-pruning proves $\rk(T_2)\ge 5$.
\end{itemize}

These results hold even if we permute the axes of $T_1,\ T_2$ arbitrarily before applying the pruners.
These examples were found by computer search.
\end{remark}

Other pruners are possible, some of which we list here for the case $r=n_0$; but we did not use them, either because they were too slow or because they did not seem to help in finding CPDs for some hand-picked tensors.
We leave it as an exercise for the reader to prove their correctness and to extend them to general $r$.

In all examples, $T$ is concise and has shape $n_0\times\dots\times n_{D-1}$, and each pruner is expressed as a condition that would have to be satisfied if $\rk(T)\le R$.

\begin{lemma}[``$k$-th order rref-pruning"]
\label{k-th-rref-pruning}
For any $1\le k\le n_0$, the set
\[
\brace{v_0\wedge\dots\wedge v_{k-1} : \ v_0,\dots,v_{k-1}\in\F^{1\times n_0},\ \rk\paren{\MA{c}{v_0\\\hline \vdots\\\hline v_{k-1}}\times_0 T}\le R-n_0+k}\subseteq \F^{\overbrace{n_0\times\dots\times n_0}^k}
\]
spans all of $\F^{\overbrace{n_0\times\dots\times n_0}^k}$,
where $v_0\wedge\dots\wedge v_{k-1} := \sum_{\sigma\in\operatorname{Sym}(k)} \operatorname{sgn}(\sigma) \bigotimes_{0\le i<k} v_{\sigma(i)}$ is (isomorphic to) the wedge product in exterior algebra.
\end{lemma}

\begin{lemma}[``frequency-pruning"]
For any $1\le k\le n_0$, $\abs{\brace{
v\in\F^{1\times n_0}:\ 
\rk(v\times_0 T)\le R-n_0+k
}}\ge |\F|^k$.
\end{lemma}

\subsection{Border CPD}
\label{sec:border}
We now present our search algorithm for border CPD, a looser form of CPD that is essential for most progress on fast matrix multiplication. Most of the text in this subsection is duplicated from \cite{yang24}.

\subsubsection{Motivation and Definitions}
Intuitively, a border CPD is a parameterized CPD that approximates a target tensor $T$ arbitrarily well. Such a CPD may have a strictly smaller rank than $\rk\paren{T}$: a famous example is $T=\M{\M{0&1\\1&0}&\M{1&0\\0&0}}$, which has rank 3 but is approximated by the rank-2 border CPD $\frac{1}{x}\paren{\M{1\\x}^{\otimes 3} - \M{1\\0}^{\otimes 3}}$ as $x\rightarrow 0$.


To formally define border CPDs, we multiply everything by a high enough power of $x$ so that everything is a polynomial in $x$. Define the following for a tensor with elements in $\F$ \cite{fmm-survey}:
\begin{itemize}
    \item $\brk{H}(T) := \rk(x^{H-1} T)$ over the polynomial ring $\F[x]/(x^H)$; we call $H$ the \textit{exponent threshold}.
    \item $\borderk{T} := \min_{H\ge 1} \brk{H}(T)$; this is the \textit{border rank} of $T$.
\end{itemize}

Border rank is useful in fast MM because one can essentially pretend it is regular rank when bounding the running time exponent $\omega$: a border rank-$R$ CPD of $\ang{m,n,k}$ for any exponent threshold $H$ can be converted into an $O(N^{(3\log_{mkn} R)+\varepsilon})$-time algorithm for $N\times N$ matrix multiplication, for arbitrarily small $\varepsilon>0$ \cite{fmm-survey}.

\subsubsection{Preliminaries}
To make the border CPD search problem finite, we fix the exponent threshold $H$. We also generalize the problem to finding the CPD rank of an arbitrary tensor over the ring $\F[x]/(x^H)$, i.e., tensor elements are not restricted to $\F$-multiples of $x^{H-1}$.

For brevity, we denote $\F[x]/(x^H)$ as $\borderring{\F}{H}$ and call it the \textit{border ring of $\F$ with exponent threshold $H$}.

We first construct a border analogue of row reduction exists on matrices, which allows us to assume tensor conciseness WLOG in border CPD.

\begin{lemma}
An element $\alpha\in\borderring{\F}{H}$ has a multiplicative inverse if and only if it is not a multiple of $x$. Furthermore, such an inverse is unique and can be found in $O(H^2)$ time.
\end{lemma}
\begin{proof}
We can express $\alpha$ as $\sum_{0\le h<H} \alpha_h x^h$, for $\alpha_h\in\F$.
Clearly, if $\alpha_0=0$ then there is no $\beta\in\borderring{\F}{H}$ such that $\beta\alpha=1$.
Otherwise, we can set $\beta=\sum_{0\le h<H} \beta_h x^h$, where $\beta_0=\frac{1}{\alpha_0}$ and $\forall h\ge 1: \beta_h=-\frac{1}{\alpha_0}\paren{\sum_{h'<h} \alpha_{h-h'}\beta_{h'}}$, so that $\beta\alpha=1$.
\end{proof}

\begin{lemma}
\label{border-reduc}
Given any $M\in\borderring{\F}{H}^{m\times n}$, there exists an invertible matrix $Q$ and a permutation matrix $P$ such that $QMP=\MA{c|c}{U&V}$, where
\begin{itemize}
    \item $U$ is square upper-triangular, and $V$ is arbitrary
    
    \item the diagonal of $U$ contains elements that are nondecreasing powers of $x$
    
    (including 0, which is equal to $x^H$ and considered a higher power of $x$ than any $x^p$ for $p<H$)
    
    \item $QMP$ contains exactly $\rk(M)$ many nonzero rows, which consist of the topmost rows

    \item $Q,P$ can be found in polynomial time
\end{itemize}
\end{lemma}
\begin{proof}
We sketch a procedure for reducing $M$ by applying row operations and column permutations:
\begin{itemize}
    \item if $M=O$, stop;
    \item otherwise, find some $(i_0,j_0)$ such that $M_{i_0,j_0}\not\in x\borderring{\F}{H}$;
    \item if such $(i_0,j_0)$ exists:
    \begin{itemize}
        \item swap rows 0 and $i_0$, and columns 0 and $j_0$,
        so that this element is moved to $(0,0)$
        \item multiply row 0 by an inverse of $M_{0,0}$, so that $M_{0,0}=1$
        \item add scalar multiples of row 0 to all other rows so that $M_{1:,j_0}=\vec{0}$
        \item recursively apply the procedure on $M_{1:,1:}$
    \end{itemize}
    
    \item otherwise: we must have $M=xM'$ for some $M'\in\borderring{\F}{H}^{m\times n}$;
    
    apply the procedure on $M'$, tracking the list of row and column operations we apply; then apply those operations on $M$
\end{itemize}

The corresponding matrices $Q, P$ can be constructed and updated as we apply each row and column operation. It is clear from the procedure that $QMP$ has the structure

\[\M{
\begin{array}{ccccc|ccc}
    x^{p_0}&y_{0,1}&y_{0,2}&\cdots&y_{0,r-1} & y_{0,r}&\cdots&y_{0,n-1}\\
    &x^{p_1}&y_{1,1}&\cdots&y_{1,r-1} & y_{1,r}&\cdots&y_{1,n-1}\\
    &&x^{p_2}&\cdots&y_{2,r-1} & y_{2,r}&\cdots&y_{2,n-1}\\
    &&&\ddots&\vdots & \vdots && \vdots\\
    &&&&x^{p_{r-1}} & y_{r-1,r}&\cdots&y_{r-1,n-1}
\end{array}
\\\hline
O
},\]

where $r$ is the number of nonzero rows in $QMP$, each $p_i<H$ is a nonnegative integer, $p_0\le \dots \le p_{r-1}$, and each $y_{i,j}$ is an element in $\borderring{\F}{H}$.
Additionally, each $y_{i,j}$ is a multiple of $x^{p_i}$, since at each step of the reduction process, $p_i$ is the smallest number such that the current state of $M_{i:,:}$ contained some element $\not\in x^{p_i+1}\borderring{\F}{H}$. As a consequence, for each $i$, $x^{p_i}$ divides all elements in $(QMP)_{i:,:}$.

To prove $r$ is indeed equal to $\rk(M)$, we will apply additional column operations on $QMP$ (which are not necessarily column operations).
Because of hte aforementioned divisibility properties of $x^{p_i}$,
it is possible to first add scalar multiples of column 0 to all other columns such that all $y_{0,:}$ are eliminated, then repeat for column 1, etc., so that we end up with the matrix

\[\MA{c|c}{
\begin{array}{ccc}
    x^{p_0}\\
    &\ddots\\
    &&x^{p_{r-1}}
\end{array}
& O
\\\hline
O & O
}.\]

Clearly, the rank of this matrix is at most $r$ (since there are $r$ many nonzero terms) and at least $\rk(x^{H-1} I_r)$ (since we can discard all-zero rows and columns, then multiply on the left by a suitable diagonal matrix).

We present a proof from Austin Conner \cite{austin-conner} that $\rk(x^{H-1} I_r)=r$.
Suppose not: then there would exist matrices $U\in\borderring{\F}{H}^{r\times r'},\ V\in\borderring{\F}{H}^{r'\times r}$, for some $r'<r$, such that $x^{H-1} I_r=UV$.
In the larger polynomial ring $\F[x]$, this condition is equivalent to there existing some $M\in\F[x]^{r\times r}$ such that $x^{H-1}I_r + x^H M=UV$.
The left-hand side has a polynomial determinant with coefficient 1 at the $x^{r(H-1)}$ term, but the right-hand side has determinant 0 because it is rank-deficient; contradiction.
\end{proof}

\subsubsection{Border CPD algorithm}
Lemma \ref{border-reduc} gives a border analogue of row-echelon form for a matrix.
As a consequence, any tensor $T$ can be transformed with invertible axis-operations to be concise, in polynomial time.
Using the same argument as in Section \ref{sec:cpd}, a concise tensor $T$ with rank $\le R$ must have its length along each axis be $\le R$.

Unfortunately, we cannot always obtain a border analogue of \textit{reduced} row-echelon form (rref), i.e. force the matrix $U$ in Lemma \ref{border-reduc} to be diagonal, using only row operations and column permutations, because some terms on the diagonal of $U$ may be non-invertible and thus unable to eliminate elements directly above them. Consider $M=\M{1&1\\&x}$; there is no scalar multiple of $M_{1,:}$ that can be added to $M_{0,:}$ such that $M_{0,1}$ is eliminated.

It is important to emphasize that when performing axis-operations on a tensor, we are not allowed to use general column operations. Doing an axis-$d$ operation on a tensor with CPD $\cpdeval{A_0,\dots,A_{D-1}}$ is equivalent to replacing the factor matrix $A_d$ with $MA_d$ (left multiplication), for some matrix $M$, which performs a row operation on $A_d$ (ignoring invertibility).
However, replacing $A_d$ with $A_d M$ (right multiplication) changes the resulting tensor in a way that cannot be determined without knowing the individual factor matrices.
Column \textit{permutations} are allowed though, because if $P$ is a permutation matrix, then $\cpdeval{A_0 P,\dots,A_{D-1}P}=\cpdeval{A_0,\dots,A_{D-1}}$ due to commutativity of addition.

The lack of border rref means that when we attempt to use the same algorithm as in CPD over a field, the axis-0 slices of $\~{Q}\times_0 \~{T}$ mixes together the outer products $\bigotimes_{d\ge 1} (A_d)_{:,i}$ for multiple $i$, so we cannot treat these slices independently.
Instead, our border CPD search algorithm relies on a simpler idea:

\begin{proof}[Proof of Theorem \ref{thm:border}]
If $\rk(T)\le R$ (and $R>0$), then there exist vectors $u_d\in\borderring{\F}{H}^{n_d},\ 0\le d<D$, such that $\rk\paren{T-\bigotimes_d u_d}\le R-1$. We can thus determine whether $\rk(T)\le R$ by enumerating all possible tuples $(u_d)_d$ and recursively solving the corresponding subproblem (while transforming the tensor in each subproblem to be concise).
The case $R=0$ is trivial.

To bound the running time of this algorithm, a subproblem with inputs $(T'\in\F^{n'_0\times\dots\times n'_{D-1}},\ r)$ will have at most $|\F|^{H\sum_d n'_d}$ children, all of the form $(\dots,\ r-1)$.
WLOG we can assume in each subproblem that $T'$ is concise and $\forall d:\ n'_d\le r$.
Thus, the total runtime of deciding whether $\rk(T)\le R$ is $\borderResult$,
where $n_0\times\dots\times n_{D-1}$ denotes the shape of the original tensor $T$ (after being transformed to be concise).
\end{proof}

Pseudocode for our border CPD algorithm is provided in Algorithm \ref{alg:border}.

\begin{algorithm}
    \caption{Depth-first search algorithm for rank-$\le R$ CPD over border ring $\F[x]/(x^H)$}
    \label{alg:border}
    \begin{algorithmic}[1]
    \newcommand{\target}{\textrm{target}}
    \Require{
        $T_\target\in\paren{\F[x]/(x^H)}^{n_0\times\dots\times n_{D-1}}$,
        $R\in\Z_{\ge 0}$
    }
    \Ensure{
        returns a rank-$\le R$ CPD of $T_\target$ as a sequence of factor matrices, if one exists;
        else, returns null
    }
    \Function{dfs}{$T_\target, R$}
        \Let{$T, Q_0,\dots,Q_{D-1}, r_0,\dots,r_{D-1}$}{\Call{border\_concise}{$T_\target$}} \Comment{satisfies $T_\target=Q_0\times_0 (\dots (Q_{D-1}\times_{D-1} T)\dots )$, $T$ concise and shape $r_0\times\dots\times r_{D-1}$}
        \If{$\exists d | r_d>R$}
            \State \Return{\textbf{null}}
        \EndIf
        \If{$\exists d | r_d=0$}
            \State \Return{$\paren{\M{\ }}_{0\le d<D}$}
        \EndIf
        \For{$0\le d<D,\ u_d\in\paren{\F[x]/(x^H)}^{r_d}$}
            \Let{$A$}{\Call{dfs}{$T-\bigotimes_d u_d ,\ R-1$}}
            \If{$A\ne\textbf{null}$}
                \State \Return{ $\paren{\paren{Q_d^{-1}}_{:,:r_d} \MA{c|c}{A_d & u_d}}_{0\le d<D}$ }
            \EndIf
        \EndFor
      \State \Return{\textbf{null}}
    \EndFunction
  \end{algorithmic}
\end{algorithm}

\subsection{Barriers}
We note some limitations to our algorithms for ordinary and border CPD, and to our pruners. Unless otherwise stated, assume $T$ is concise, has shape $n_0\times\dots\times n_{D-1}$, and satisfies $n_0\ge\dots\ge n_{D-1}$.

The simplest nontrivial rank threshold that our ordinary CPD algorithm can solve is $R=n_0$, since there are no variables from any $A_d$ that have to be fixed. In this case, our algorithm runs in $O^*(|\F|^{n_0})$ time and is equivalent to running both rref-pruning (Lemma \ref{rref-prune}) and rref-heuristic.

In hopes of generalizing this result, we tried analyzing the rank threshold $R=n_0+1$ using $k$-th order rref-pruning.
Unfortunately, the following tensor breaks this idea for $k=2$:

\begin{example}
\label{2nd-order-rref-counterexample}
Over the field $\F_2$, the tensor
$T=\M{\M{1&1\\&1\\&&\phantom{0}}&\M{\phantom{0}\\&1&1\\&&1}&\M{1\\&\phantom{0}\\1&&1}}$
satisfies both rref-pruning (Lemma \ref{rref-prune}) and 2nd-order rref-pruning (Lemma \ref{k-th-rref-pruning}, $k=2$) for rank threshold $R=4$ and no variables assigned ($r=n_0$), along all axes simultaneously.
In fact, for any two distinct indices $i$ and $j$, $\rk\paren{\M{T_{i,:,:} & T_{j,:,:}}}\le 3$, and the same property holds after rotating the axes of $T$.

However, running Algorithm \ref{alg:main} shows that $\rk_{\F_2}\paren{T}=5$.
\end{example}
\begin{remark}
The tensor in Example \ref{2nd-order-rref-counterexample} also passes lask-pruning with rank threshold $R=4$ (no variables assigned) on all axes.
\end{remark}
\begin{remark}
The axis-0 slices of the tensor in Example \ref{2nd-order-rref-counterexample} are diagonal shifts of each other.
Explicitly, $T=\M{A & PAP^\intercal & P^2 A (P^2)^\intercal}$, where $A=\M{1&1\\&1\\&&\phantom{0}}$ and $P=\M{&&1\\1\\&1}$.
\end{remark}

Example \ref{2nd-order-rref-counterexample} was found by first running a computer search, then transforming the resulting tensor so that $k$-th order rref pruning is easier to visually check.

We are not sure if a counterexample exists for $k=3$ or higher, although we suspect one does:
\begin{problem}
Does there exist a 3-dimensional tensor $T\in\F^{n_0\times n_1\times n_2}$ such that:
\begin{itemize}
    \item $T$ passes $k$-th order rref-pruning (Lemma \ref{k-th-rref-pruning}) for all $1\le k\le 3$ and rank threshold $R:=n_0+1$, along all axes;
    \item but $\rk(T)>n_0+1$?
\end{itemize}
\end{problem}

Another algorithmic barrier is that rref-pruning and lask-pruning are unable to prove superlinear lower bounds on tensor rank.
Specifically, when the input tensor $T$ has shape $n_0\times n_1\times n_2$ (with $n_0\ge n_1\ge n_2$) and no variables have been assigned ($r=n_0$):
\begin{itemize}
    \item rref-pruning cannot disprove the statement $\rk(T)\le R$ when $R\ge n_1+n_0$, because $\rk(v\times_0 T')$ is guaranteed to have rank $\le n_1$ just by considering matrix shape, so $S'$ will contain all of $\F^{1\times n_0}$
    (refer to Lemma \ref{rref-prune} for notation);
    \item and lask-pruning cannot disprove $\rk(T)\le R$ when $R\ge \frac{|\F|^{n_0}\cdot n_1}{|\F|^{n_0}-|\F|^{n_0-1}}
    =\frac{|\F|}{|\F|-1}n_1$, since the sum $s$ is at most $|\F|^{n_0}\cdot n_1$
    (refer to Lemma \ref{lask-prune} for notation).
\end{itemize}

In fact, providing any algorithm that is capable of proving superlinear lower bounds on tensor rank for 3-dimensional tensors is an open problem (see Problem \ref{beyond-thm:sub-method}). We suspect the other pruners we have listed are also unable to prove superlinear lower bounds (without becoming impractically slow, e.g., running $k$-th order rref-pruning for high $k$).

The situation for border CPD is much more difficult, as the extra multiplicative $H$ term in the exponent of Theorem \ref{thm:border} renders our search algorithm prohibitively slow for most inputs beyond tensor shape $3\times 3\times 3$, rank threshold $R=3$, exponent threshold $H=2$, and ground field $\F=\F_2$ (substituting this into the asymptotic running time expression yields $2^{36}$).
Additionally, the lack of reduced row-echelon form makes it difficult to devise efficient pruners similar to those for ordinary CPD search.

One potential area for improving our border CPD search, as mentioned in \cite{yang24}, is to analyze the case when the rank threshold is close to $n_0$ (in the extreme case, $R=n_0$). We suspect that most recursive function calls are terminated immediately because the input tensor is too long along some axis after being made concise (Algorithm \ref{alg:border}, line 3), so preventing the algorithm from making these calls in the first place may reduce asymptotic time complexity.

\section{Maximum rank}
\label{sec:max-rank}
In this section we study tensors more generally, beyond the context of fast matrix multiplication.
One problem of interest is finding the maximum possible rank that a tensor of a given shape $n_0\times\dots\times n_{D-1}$ can have over a field $\F$, which we denote $\maxrank_\F(n_0,\dots,n_{D-1})$. We omit $\F$ if it is clear from context.
Outside of this introduction, we only study 3-dimensional tensors ($D=3$).

By permuting axes, we immediately have that $\maxrank(n_0,\dots,n_{D-1})=\maxrank(n_{\sigma_0},\dots,n_{\sigma_{D-1}})$ for any permutation $\sigma\in\operatorname{Sym}(D)$.
Additionally, the 2-dimensional case is solved by Gaussian elimination: $\maxrank(m,n)=\min(m,n)$.

By a simple counting argument, the following lower bound holds over any finite field $\F$; it also holds over infinite $\F$ is we count degrees of freedom instead of number of tensors.
\begin{theorem}
\label{counting-lower-bound}
$\maxrank(n_0,\dots,n_{D-1})\ge \frac{\prod_d n_d}{\sum_d n_d}$.
\end{theorem}

In the other direction, we can construct a CPD of any tensor by choosing some axis $d$ and constructing each 1D slice $(i_0,\dots,i_{d-1},:,i_{d+1},\dots,i_{D-1})$ at a time, for each tuple of coordinates $(i_{d'})_{d'\ne d}$. This construction proves the following upper bound:
\begin{theorem}
\label{trivial-upper-bound}
$\maxrank(n_0,\dots,n_{D-1})\le \min_d \prod_{d'\ne d} n_{d'}$.
\end{theorem}

As an immediate consequence, for 3-dimensional cube-shaped tensors we have $\frac{n^2}{3}\le \maxrank(n,n,n)\le n^2$. In fact, over a finite field $\F$, a uniformly random tensor will have rank $\ge \frac{n^2}{3}$ with high probability.
Unfortunately, we do not know an \textit{explicit} (i.e. deterministic) tensor with rank nearly as high; the following is a major open problem:
\begin{problem}[\cite{blaser-explicit-tensors}]
\label{beyond-thm:sub-method}
Is there an explicit tensor $T\in \F^{n\times n\times n}$ satisfying $\rk(T)\ge (3+\varepsilon)n$ for any constant $\varepsilon>0$?
\end{problem}

The strongest known lower bound on $\maxrank(n,n,n)$ from an explicit tensor is $3n-o(n)$ (Example \ref{landsberg-michalek-3n}, Remark \ref{landsberg-mm-lower}).

The main reason why most values $\maxrank(n_0,\dots,n_{D-1})$ are difficult to find is because of a lack of symmetry:
the set of all $n_0\times \dots\times n_{D-1}$ tensors has size $|\F|^{\prod_d n_d}$, or $\prod_d n_d$ degrees of freedom (DoF);
but the group of all invertible axis-operations on the tensor shape, $\GL{n_0}{\F}\times\dots\times \GL{n_{D-1}}{\F}$, has size $\Theta(|\F|^{\sum_d n_d^2})$ (for constant $D$), or only $\sum_d n_d^2$ degrees of freedom \cite{cohn-stackexchange}.
When $D=2$, the latter has more DoF than the former, so there exists a simple canonical form of matrices up to invertible axis-operations (obtainable by applying Gaussian elimination on the rows, and then on the columns); but for most shapes with $D\ge 3$, there is no such form.

This problem with degrees of freedom also explains why $\maxrank(n,n,2)$ has been completely solved (see Theorem \ref{nn2-max-rank}) but $\maxrank(n,n,3)$ is still not well-understood: the former has $2n^2$ DoF for tensors, which is comparable to $2n^2+2^2$ for invertible axis-operations; but the latter has $3n^2$ versus $2n^2+3^2$.

\subsection{Previous work}
Theorem \ref{counting-lower-bound} is the strongest known lower bound on maximum rank when the tensor shape is roughly cubic. For unbalanced shapes, stronger results are known.

The shape $m\times n\times 2$ is completely solved:
\begin{theorem}[\cite{nn2-max-rank}]
\label{nn2-max-rank}
For a finite field $\F$ and integers $m\ge n\ge 2$, the following holds:

if $\F=\F_2$: $\maxrank_{\F_2}(m,n,2)=\cases{n+\ceil{m/2} & n\le m\le 2n-2 \\ 2n-1 & m=2n-1 \\ 2n & m\ge 2n}$;

otherwise: $\maxrank_\F(m,n,2)=\cases{n+\floor{m/2} & n\le m\le 2n-1 \\ 2n & m\ge 2n}$.
\end{theorem}

For skinny shapes, i.e. $m\times n\times (mn-k)$ for small $k$, the following is known:
\begin{theorem}[\cite{bshouty}]
For $0\le k\le \min\paren{\frac{(n-1)^2}{2},\ \frac{(m-1)^2}{2}}$: $\maxrank(m,n,mn-k)\ge mn-4\sqrt{2k}+O(1)$.
\end{theorem}
We improve the constant factor on $\sqrt{k}$ from $4\sqrt{2}$ to $3\sqrt{3}$ in Theorem \ref{3-sqrt3}, using the same idea as \cite{bshouty} used for this result.

For upper bounds, the following result is known over any \textit{principal ideal domain}; for brevity, we state the proof over a field:
\begin{theorem}[\cite{howell}]
\label{howell}
$\maxrank(m,n,p)\le \maxrank(m-1,n-1,p)+p$.
\end{theorem}
\begin{proof}
By permuting and scaling slices of a tensor $T\in\F^{m\times n\times p}$, WLOG we can assume $T_{0,0,0}=1$; then we can add scalar multiples of $T_{0,:,:}$ to all other $T_{i,:,:}$ so that for all $i$, $T_{i,0,0}=1$; doing so is equivalent to performing $\M{1\\\lambda_1 & 1 \\ \vdots && \ddots \\ \lambda_{m-1} &&& 1}\times_0$ for some scalars $\lambda_i$ that we can choose.

Now for each $i$, we subtract the matrix outer product $T_{i,:,0}\otimes T_{i,0,:}$ from $T_{i,:,:}$; equivalently, we subtract $e_i\otimes T_{i,:,0}\otimes T_{i,0,:}$ from $T$, where $e_i$ has a 1 at index $i$ and zeros everywhere else.
Doing so changes each slice $T_{i,:,:}$ into the form $\bracket{\begin{array}{c|c}
0 & \begin{array}{ccc}0&\dots&0\end{array} \\
\hline
\begin{array}{c}0\\\vdots\\0\end{array} & M
\end{array}}$ for some arbitrary $M$.
Finally, removing the zeros at slices $(:,0,:)$ and $(:,:,0)$ leaves us with a $(m-1)\times (n-1)\times p$ tensor.
\end{proof}

\begin{corollary}
\label{howell-nnn}
$\maxrank(n,n,n)\le \ceil{\frac{3}{4}n^2}$.
\end{corollary}
\begin{proof}
Apply Theorem \ref{howell} $\floor{\frac{n}{2}}$ many times on $\maxrank(n,n,n)$ to get $\maxrank(n,n,n)\le \maxrank(\ceil{\frac{n}{2}},\ceil{\frac{n}{2}},n)+n\floor{\frac{n}{2}}$. Then by Theorem \ref{trivial-upper-bound},
$\maxrank(\ceil{\frac{n}{2}},\ceil{\frac{n}{2}},n)\le \ceil{\frac{n}{2}}^2$, so $\maxrank(n,n,n)\le \ceil{\frac{n}{2}}^2+n\floor{\frac{n}{2}}$. Some even-odd casework shows that this expression equals $\ceil{\frac{3}{4}n^2}$.
\end{proof}

Over the complex numbers $\C$, a much stronger and better-known result holds:
\begin{theorem}[\cite{atkinson-stephens}]
\label{atkinson-stephens}
$\maxrank_\C(m,n,p)\le m+\floor{p/2}n$.
\end{theorem}

In Section \ref{sec:max-rank-barriers}, we show that the technique used by \cite{atkinson-stephens} to prove Theorem \ref{atkinson-stephens} fails over a finite field.

Finally, exact values of maximum rank are known for small tensor shapes:
\begin{itemize}
    \item $\maxrank(2,2,2)=3$ over any field (see Example \ref{W-state} and Theorem \ref{howell-nnn})
    \item $\maxrank(3,3,3)\le 6$ over any field, and for some fields (e.g. $\F_2$) this bound is tight \cite{lavrauw-pavan-zanella, bremner-hu}
    \item $\maxrank_\C(3,3,3)=5$ \cite{bremner-hu}
\end{itemize}

We give more details and new results in Section \ref{max-rank-new-results}.

\subsection{Substitution method}
To find new results, as well as restate previous ones in a clean form, we will often want to prove lower bounds on the tensor rank of specific tensors. One powerful method to do so is the \textit{substitution method}, first described by \cite{substitution-method} in the context of polynomials and repurposed by \cite{alexeev-forbes-tsimerman} for tensors; we follow the latter source. Intuitively, if we remove a slice along an arbitrary axis from a tensor, there exists a way of adding scalar multiples of that slice to the remaining slices so that the tensor rank decreases by at least one.

\begin{theorem}[\cite{alexeev-forbes-tsimerman}, Appendix B]
\label{thm:sub-method}
Let $T\in\F^{n_0\times\dots\times n_{D-1}}$ be a tensor, $d$ be an axis, and $0\le i<n_d$ be a coordinate such that the slice $T_{\cdots,:,i,:,\cdots}$ along axis $d$ is nonzero. Then there exist scalars $\lambda_j$ for each $0\le j<n_d,\ j\ne i$ such that the tensor $T'$ with axis-$d$ slices $T_{\cdots,:,0,:,\cdots}+\lambda_0 T_{\cdots,:,i,:,\cdots},\ \cdots,\ T_{\cdots,:,n_d-1,:,\cdots}+\lambda_{n_d-1} T_{\cdots,:,i,:,\cdots}$ satisfies $\rk(T')\le \rk(T)-1$.
\end{theorem}
\begin{proof}
By permuting axes and slices, it suffices to prove the case $d=0,\ i=0$; then $T'=\M{\lambda_1 & 1 \\ \vdots && \ddots \\ \lambda_{n_0-1} &&& 1} \times_0 T$.

Let $T=\cpdeval{A_0,\dots,A_{D-1}}$ be a rank-$\rk(T)$ CPD. Because $T_{0,:,\cdots}\ne O$, there must exist some $r$ s.t. $(A_0)_{0,r}\ne 0$. Then setting $\lambda_j=-\frac{(A_0)_{j,r}}{(A_0)_{0,r}}$ forces the $r$-th column of $\M{\lambda_1 & 1 \\ \vdots && \ddots \\ \lambda_{n_0-1} &&& 1}A_0$ to be all-zeros.
Since $T'=\cpdeval{\M{\lambda_1 & 1 \\ \vdots && \ddots \\ \lambda_{n_0-1} &&& 1}A_0, A_1,\dots,A_{D-1}}$, removing the $r$-th column of each factor matrix from this CPD does not change the resulting tensor, so $\rk(T')\le \rk(T)-1$.
\end{proof}

\begin{corollary}
Using the same $T$ and $T'$ as in Theorem \ref{thm:sub-method}, if we can find some $R$ such that $\rk(T')\ge R$ for all choices of scalars $\lambda_j$, then $\rk(T)\ge R+1$.
\end{corollary}

When applying the substitution method to the slice $T_{\cdots,:,i,:,\cdots}$ along some axis $d$, we call this action ``subbing out $(\cdots,:,i,:,\cdots)$".

Overall, the substitution method is the strongest method of proving tensor rank lower bounds without resorting to computational search \cite{blaser-explicit-tensors}. Because using the substitution method to its full potential may require subbing out along every axis (see Example \ref{landsberg-michalek-3n}), it is difficult to combine it with our CPD search algorithm (Theorem \ref{thm:main}).

Below we list some examples of the substitution method, including previously discovered results:

\begin{example}
\label{W-state}
The ``W-state" tensor $T=\M{\M{1\\&\phantom{0}}&\M{&1\\1}}$ has rank 3 over any field.
\end{example}
\begin{proof}
Subbing out $(:,1,:)$ yields $\M{\M{1&\phantom{0}}&\M{x&1}}$ for some scalar $x$ we cannot control.
Then, subbing out $(:,:,1)$ yields $\M{\M{1}&\M{y}}$ for a (possibly different) scalar $y$. Finally, subbing out $(:,:,0)$ yields $\M{}$ (an empty tensor), so we cannot continue.

Since an empty tensor trivially has rank 0, and we subbed out three slices from $T$, we have proven $\rk(T)\ge 3$.
On the other hand, $\rk(T)\le 3$ because there are only 3 nonzero elements in $T$.


\end{proof}

Example \ref{W-state} is an explicit tensor of shape $n_0\times n_1\times n_2$ with rank strictly higher than any $n_d$, which is not possible for matrices.

The next example combines the substitution method with axis-operations:

\begin{example}[Thm, \cite{brockett-dobkin}]
\label{mm222-rank7}
The tensor $T=\ang{2,2,2}$ corresponding to matrix multiplication of two $2\times 2$ matrices has rank 7 over any field.
\end{example}
\begin{proof}
We can permute slices along several axes so that $T$ becomes
\[\M{
\M{1\\&\phantom{0}\\&&1\\&&&\phantom{0}}
&
\M{&1\\\phantom{0}\\&&&1\\&&&\phantom{0}}
&
\M{&&&\phantom{0}\\1\\&\phantom{0}\\&&1}
&
\M{\phantom{0}\\&1\\\phantom{0}\\&&&1}
}.\]

Subbing out $(:,3,:)$ yields
\[\M{
\M{1\\&\phantom{0}\\&&1&\phantom{0}}
&
\M{&1\\\phantom{0}\\&&&1}
&
\M{&&x&\phantom{0}\\1&&y\\&&z}
&
\M{\phantom{0}&&&x\\&1&&y\\&&&z}
}\]

for some scalars $x,y,z$ we cannot control.

We notice that the sub-tensor $(:2,:,:)$ has rank 4 (because its axis-2 slices are linearly independent). To prove $\rk(T)\ge 7$, we would like to sub out slices $(2,:,:)$ and $(3,:,:)$, and then prove the resulting tensor has rank $\ge 4$. However, if we did that right now, we would destroy the 1s at $(0,2,2)$ and $(1,2,3)$, and the resulting tensor would at best only have rank 2.

To avoid destroying terms in $(:2,:,:)$, we perform some invertible axis-operations (which do not affect tensor rank) to cancel some unwanted terms in $(2,:,:)$ and $(3,:,:)$. First, we apply $\M{1\\&1\\-z&&1\\&-z&&1}\times_0$ to get

\[\M{
\M{1\\&\phantom{0}\\&&1&\phantom{0}}
&
\M{&1\\\phantom{0}\\&&&1}
&
\M{-z&&x&\phantom{0}\\1&&y\\&&}
&
\M{\phantom{0}&-z&&x\\&1&&y\\&&&}
};\]

then we apply $\M{1&z\\&1\\&&1}\times_1$ to get

\[\M{
\M{1\\&\phantom{0}\\&&1&\phantom{0}}
&
\M{&1\\\phantom{0}\\&&&1}
&
\M{&&x+yz&\phantom{0}\\1&&y\\&&}
&
\M{\phantom{0}&&&x+yz\\&1&&y\\&&&}
}.\]

At this point, we replace all indeterminates with asterisks to simplify notation, as we will no longer have to examine their exact structure.

Now we sub out slices $(2,:,:)$ and $(3,:,:)$ to get


\[\scalebox{0.9}{$\M{
\M{1\\&\phantom{0}\\&&1&\phantom{0}}
+(*)\M{&&*&\phantom{0}\\ *&&*\\&&}
+(*)\M{\phantom{0}&&& *\\& *&& *\\&&&}
&\phantom{0}&
\M{&1\\\phantom{0}\\&&&1}
+(*)\M{&&*&\phantom{0}\\ *&&*\\&&}
+(*)\M{\phantom{0}&&& *\\& *&& *\\&&&}
}$}\]
\[
=\M{
\M{1&&*&*\\ *&*&*&*\\&&1&\phantom{0}}
&
\M{&1&*&*\\ *&*&*&*\\&&&1}
}.
\]

This tensor has rank at least 4, no matter what each * is replaced with,
since we can sub out $(:,:,0)$, $(:,:,1)$, $(:,:,2)$, and $(:,:,3)$ in that order; each slice we sub out avoids destroying the 1s in later slices.

Since we sub out seven slices in total,
$\rk(T)\ge 7$.
Because we did not perform any scalar divisions, this proof works over any field.

On the other hand, $\rk(T)\le 7$ over any field, because the CPD corresponding to Strassen's algorithm only contains integers.
\end{proof}

\begin{remark}
An earlier proof of Example \ref{mm222-rank7} is from \cite{winograd-222} and uses matrix determinants.
\end{remark}

The substitution method proof of Example \ref{mm222-rank7} can be generalized to the following:
\begin{theorem}[\cite{brockett-dobkin}]
$\rk(\ang{m,n,2})\ge mn+m+n-1$ over any field.
\end{theorem}

However, the substitution method has its limits. Because subbing out one slice contributes 1 point to the lower bound, we cannot obtain a lower bound better than $\sum_d n_d$.
Proving a lower bound even slightly above this limit is essentially the motivation for Problem \ref{beyond-thm:sub-method} and will likely require nonlinear methods.

The next few examples use the substitution method close to its limit:
\begin{example}[\cite{abelian-tensors}]
\label{landsberg-michalek-2n}
The tensor \[T(k):=\M{\M{1\\&\phantom{0}\\&&\phantom{0}\\&&&\phantom{0}} & \M{&1\\1\\&&\phantom{0}\\&&&\phantom{0}} & \M{&&&1\\&&1\\&1\\1}&\dots}\]

with shape \footnote{In \cite{abelian-tensors}, the tensor shape was incorrectly written as $k\times 2^k\times 2^k$.} $(k+1)\times 2^k\times 2^k$ has rank $2\cdot 2^k-1$.
\end{example}
\begin{proof}
Sub out slice $(:,i,:)$ for each $i=2^k-1,\dots,2^{k-1}$; then slice $(:,:,i)$ for each $i=2^k-1,\dots,2^{k-1}$; then remove slice $(k-1,:,:)$. Afterwards, we have the tensor $T(k-1)$, so we have shown $\rk(T(k))\ge \rk(T(k-1))+2^k$.
Applying induction on $k$ proves $\rk(T(k))\ge 2\cdot 2^k-1$, and counting the number of nonzero terms in $T(k)$ proves a matching upper bound.
\end{proof}

\begin{example}[\cite{abelian-tensors}]
\label{landsberg-michalek-3n}
The $2^k\times 2^k\times 2^k$ tensor $T'(k)$ such that
\begin{itemize}
    \item $T'(k)_{:k+1,:,:}=T(k)$, using the tensor $T(k)$ from Example \ref{landsberg-michalek-2n};
    \item for $i\ge k+1$: $T'(k)_{i,:,:}$ has a 1 at element $(2^k-1,\ 2^k-1-(i-k))$ and zeros everywhere else;
\end{itemize}

has rank $\ge 3\cdot 2^k - k - 3$.
\end{example}
\begin{proof}
Sub out slice $(i,:,:)$ for each $i\ge k+1$; then remove slice $(:,2^k-1,:)$. Doing so leaves us with the truncated tensor $T(k)_{:,:2^k-1,:}$. We can then continue as in the proof for Lemma \ref{landsberg-michalek-2n}, except we sub out one less slice, so
$\rk(T(k)_{:,:2^k-1,:})\ge 2\cdot 2^k - 2$.
Thus, $\rk(T'(k))\ge \rk(T(k)_{:,:2^k-1,:})+(2^k-(k+1))
=3\cdot 2^k - k - 3$.
\end{proof}

\begin{remark}
\label{landsberg-mm-lower}
The matrix multiplication tensor $\ang{n,n,n}$ is also known to have rank $\ge 3n-o(n)$ \cite{landsberg-mm-lower}.
\end{remark}

Finally, we use Example \ref{landsberg-michalek-2n} to prove lower bounds on $\maxrank(p,n,n)$ for $p=2,3$:
\begin{theorem}
\label{diag-stack}
For any constant $p$, $\maxrank(p,n,n)\ge (2-\frac{1}{2^p})n-O(1)$ over any field.
\end{theorem}
\begin{proof}
Let $T$ be the tensor in Example \ref{landsberg-michalek-2n} with $k=p$. Consider diagonally stacking $\floor{\frac{n}{2^p}}$ many copies of $T$ as so:

$T':=\M{
\M{T_{0,:,:} \\ &\ddots \\ &&T_{0,:,:}}
&\dots&
\M{T_{k-1,:,:} \\ &\ddots \\ &&T_{k-1,:,:}}
},$

then padding $T'$ with zeros so that it has shape $p\times n\times n$.

Because we are working with a larger tensor, subbing out one slice will cause scalar multiples of it to be added to a larger set of elements; however, because of how the copies of $T$ are placed, subbing out along axes 1 or 2 will not disrupt any terms in the sub-tensors $\paren{:,\ i\cdot 2^k:(i+1)2^k,\ i\cdot 2^k:(i+1)2^k}$ for $0\le i<\floor{\frac{n}{2^p}}$.

Because the proof of Example \ref{landsberg-michalek-2n} only subbed out slices along axes 1 and 2, we can perform the same substitution procedure on each copy of $T$, one at a time, without disrupting any other copies. Thus, $\rk(T')\ge \floor{\frac{n}{2^p}}(2\cdot 2^p-1)=(2-\frac{1}{2^p})n-O(1)$ for any constant $p$.
\end{proof}

\begin{remark}
For $p=2$, Theorem \ref{diag-stack} is within an additive $O(1)$ factor of the exact value of $\maxrank(n,n,2)$ \cite{nn2-max-rank}.
\end{remark}

\begin{remark}
The case $p=3$ was previously mentioned by \cite{atkinson-stephens} without proof.
\end{remark}

\begin{remark}
Theorem \ref{diag-stack} is no better than the counting argument (Theorem \ref{counting-lower-bound}) if $p\ge 4$.
\end{remark}

\subsection{New results}
\label{max-rank-new-results}
We first improve existing bounds on $\maxrank(n,n,n)$ and $\maxrank(m,n,mn-k)$ by tuning parameters more carefully.

\begin{theorem}
\label{improved-nnn}
$\maxrank(n,n,n)\le \frac{23}{32}n^2+O(n)$ over any field.
\end{theorem}
\begin{proof}
We combine the following previous results:
\begin{itemize}
    \item Theorem \ref{howell}: $\maxrank(m,n,p)\le \maxrank(m-1,n-1,p)+p$
    \item Theorem \ref{nn2-max-rank} (weakened): $\maxrank(m,n,2)\le n+\frac{m}{2}+O(1)$ for $m\ge n$
\end{itemize}

Using the latter result, $\maxrank(m,n,p)\le \ceil{\frac{p}{2}}\paren{n+\frac{m}{2}+O(1)}
=\frac{p}{2}(n+\frac{m}{2})+O(n+m+p)$,
by splitting a $m\times n\times p$ tensor along axis 2 into $m\times n\times 2$-shaped pieces (plus a smaller remaining piece).

Consider applying Theorem \ref{howell} to $\maxrank(n,n,n)$ some $0\le k\le n$ many times, then invoking the aforementioned bound derived from Theorem \ref{nn2-max-rank}:

\[
\begin{array}{lll}
    \maxrank(n,n,n) & \le \maxrank(n-k,n-k,n)+kn \\
     & =\maxrank(n,n-k,n-k)+kn \\
     & \le \frac{n-k}{2}((n-k)+\frac{n}{2})+kn+O(n) & \textrm{if } n\le 2(n-k);
\end{array}
\]

The optimal $k$ is $\arg\min_k \frac{n-k}{2}((n-k)+\frac{n}{2})+kn+O(n)
=\arg\min_k \frac{1}{4}(2k^2-kn)
=\frac{n}{4}$, yielding $\maxrank(n,n,n)\le \frac{23}{32}n^2+O(n)$.
We check that the condition $n\le 2(n-k)=2\cdot \frac{3}{4}n=\frac{3}{2}n$ is satisfied.
\end{proof}

\begin{theorem}
\label{3-sqrt3}
For $k=o(\min(m,n)^2)$, $\maxrank(m,n,mn-k)\ge mn-3\sqrt{3k}+O(1)$ over any field.
\end{theorem}
\begin{proof}
Choose $1\le r\le m,\ 1\le s\le n$ such that $rs\ge k$, and let $T'$ be a maximal-rank $r\times s\times (rs-k)$ tensor. Then we can construct the $m\times n\times (mn-k)$ tensor $T$ such that
\begin{itemize}
    \item $T_{:,:,:(rs-k)}=T'$;
    \item for $i>rs-k$, each $T_{:,:,i}$ has a 1 at a distinct coordinate $(a,b)$ not inside $\brace{0,\dots,r-1}\times \brace{0,\dots,s-1}$, and zeros everywhere else.
\end{itemize}

We have $\rk(T)\ge \rk(T')+(mn-rs)$, since if we sub out each slice $(:,:,i)$ at a time for $i>rs-k$, and then remove slices $(r:,:,:)$ and $(:,s:,:)$, we are left with $T'$.

Thus, $\maxrank(m,n,mn-k)\ge \maxrank(r,s,rs-k)+mn-rs$.
Using the counting lower bound (Theorem \ref{counting-lower-bound}) on $\maxrank(r,s,rs-k)$ yields
\[\maxrank(m,n,mn-k)\ge \frac{rs(rs-k)}{rs+r+s-k}+mn-rs
=\frac{-rs(r+s)}{rs+r+s-k}+mn.\]

To minimize $f(r,s):=\frac{rs(r+s)}{rs+r+s-k}$, we set its partial derivatives to 0.

\[
\begin{array}{rrl}
     & \frac{\partial}{\partial r} \frac{r(r+s)}{rs+r+s-k} & =0 \\
     && =\frac{2r+s}{rs+r+s-k}-\frac{r(r+s) (s+1)}{(rs+r+s-k)^2} \\
    \Rightarrow & (2r+s)(rs+r+s-k)-r(r+s)(s+1) & =0 \\
     && =r^2s+(r+s)^2-(2r+s)k; \\
\end{array}
\]

by symmetry, having $\frac{\partial f}{\partial s}=0$ implies $s^2r+(r+s)^2-(2s+r)k=0$.

Taking the sum and difference of these two equations yields

\[
\begin{array}{rl}
    rs(r+s)+2(r+s)^2-3(r+s)k & =0 \\
    rs(r-s)-(r-s)k & =0 \\
\end{array}
\]
\[
\Rightarrow
\begin{array}{ccc}
    rs+2(r+s)-3k & =0 & \textrm{(since } r+s\ne 0 \textrm{)} \\
    (r-s)(rs-k) & =0 \\
\end{array}
\]

If $rs-k=0$, then $f(r,s)=k$. On the other hand, if $r-s=0$, the first equation implies $r^2+4r-3k=0
\Rightarrow r^*=s^*=-2+\sqrt{4+3k}$.

To simplify analysis, it suffices to plug in $r^*=s^*=\sqrt{3k}$ instead;
then $f(\sqrt{3k},\sqrt{3k})
=\frac{3k\cdot 2\sqrt{3k}}{2k+2\sqrt{3k}}
=\frac{3k\cdot \sqrt{3k}}{k+\sqrt{3k}}
=\frac{3k\cdot \sqrt{3k}}{k} \cdot \frac{k}{k+\sqrt{3k}}
=3\sqrt{3k} \cdot \paren{1-\frac{\sqrt{3k}}{k+\sqrt{3k}}}
=3\sqrt{3k} \cdot \paren{1-O\paren{\frac{1}{\sqrt{k}}}}
=3\sqrt{3k}-O(1)$.

Thus, $\maxrank(m,n,mn-k)\ge mn-3\sqrt{3k}+O(1)$.

To ensure the conditions $r^*\le m,\ s^*\le n$ are met, it suffices to have $k=o(\min(m,n)^2)$.
\end{proof}

Finally, we list some bounds on small tensor shapes that we obtained with computational search, all of which are over $\F_2$. To achieve this, we canonicalized the tensors $T$ we searched up to invertible axis-operations along two of the three axes (axes 1 and 2):
\begin{itemize}
    \item force $T_{0,:,:}=\M{I_r&O\\O&O}$ for some $r$ (i.e. simultaneously row-reduced and column-reduced);
    
    \item for $i=1,2,\dots$:
    
    force $T_{i,:,:}$ to be lexicographically minimized over all transformations $M\mapsto PMQ^\intercal$, for pairs of invertible matrices $(P,Q)$ such that $P\times_1 (Q\times_2 T_{:i,:,:})=T_{:i,:,:}$;
\end{itemize}

Additionally, when possible, we force $\rk(T_{0,:,:})\ge r_0$ for some manually chosen value $r_0$.
Specifically, if we already know $\maxrank(m,n,p)\ge R_0$ for some $R_0$, we only need to search $T$ for which there exists some $i$ such that $\rk(T_{i,:,:})\ge \floor{\frac{R_0}{m}}+1$; otherwise $T$ would immediately have rank $\le R_0$. Thus, we set $r_0=\floor{\frac{R_0}{m}}+1$.

We have solved tensor shapes $3\times 3\times 4$, $3\times 4\times 4$, and $3\times 3\times 5$; results are summarized in Table \ref{tab:max-rank-small-shapes}.
The next shape we would want to solve is $4\times 4\times 4$, but with our current method there are too many tensors to check; the best lower bound we know of is $\maxrank_{\F_2}(4,4,4)\ge 8$, which can be obtained from $\maxrank_{\F_2}(3,4,4)$ by padding.

\begin{table}[h]
    \label{tab:max-rank-small-shapes}
    \centering
    \begin{tabular}{ccc|ccc|rrcc}
        $m$ & $n$ & $p$  &  $R_0$ & example & $r_0$ & \# tensors & time (s) & $\maxrank_{\F_2}(m,n,p)$ & example \\
        && & & tensor && searched &&& tensor \\
        \hline
        3 & 3 & 4  &  6 &
        ${\scriptscriptstyle \M{&v_2&v_1\\v_2&&v_0\\v_1&v_0&v_2}}$
        & 3 & 14664 & 8.70 & \textbf{6} & rank-$R_0$ example \\
        &&  &  & (padding of &&&&&  \\
        &&  &  & $3\times 3\times 3$ &&&&&  \\
        &&  &  & from \cite{bremner-hu}) &&&&&  \\
        \hline
        3 & 4 & 4  &  7 & ${\scriptscriptstyle \M{v_0&v_1&&v_2\\v_1&&v_2\\&v_2\\v_2}}$ & 3 & 657616 & 24500 & \textbf{8} & ${\scriptscriptstyle \M{v_0&&v_2\\v_2&v_0&&v_1\\&v_2&v_0+v_1&v_1+v_2\\v_1&&v_1+v_2&}}$ \\
        &&  &  & (Ex. \ref{landsberg-michalek-2n}, $k=2$) &&&&& \\
        \hline
        3 & 3 & 5  &  6 & ${\scriptscriptstyle \M{v_0&v_1&v_2\\&v_0&v_3\\v_4}}$ & 3 & 31428 & 14.9 & \textbf{7} & ${\scriptscriptstyle \M{v_0&&&&v_2\\&v_0&v_2&&v_1\\&v_2&v_0+v_2&v_1+v_2&v_2}}$ \\
        &&  &  & ($5\times 3\times 3$) &&&&& \\
    \end{tabular}
    \caption{
        Table of maximum rank $\maxrank_{\F_2}(m,n,p)$ over $\F_2$ for specific tensor shapes $m\times n\times p$.
        $R_0$ denotes a previously known lower bound on the maximal rank and is used to reduce the search space.
        In all rows, $r_0=\floor{\frac{R_0}{m}}+1$.
        \\
        To save space, each example tensor $T$ is notated as its axis-0 contraction $v\times_0 T$ for a row vector $v$; this is sometimes called the \textit{characteristic matrix} of the tensor.
        For example, $T=\M{\M{1\\&\phantom{0}}&\M{1&2\\3}}$ is notated as $\M{v_0+v_1&2v_1 \\ 3v_1}$.
        \\
        The tensor $\M{v_0&v_1&v_2\\&v_0&v_3\\v_4}$ has rank $\ge 6$, since we can: sub out $(i,:,:)$ for each $i\ge 1$; remove $(:,2,:)$ and $(:,:,2)$;
        then do an axis-operation to transform the resulting tensor to $\M{\M{1\\&1}}$.
    }
\end{table}

\subsection{Barriers}
\label{sec:max-rank-barriers}
Here we list some difficulties in improving lower and upper bounds on maximum rank.
Besides the still-high runtime of our CPD search algorithm, finding maximum rank over finite fields will likely require field-specific techniques.

First, the best known upper bound on maximum rank over the complex numbers (Theorem \ref{atkinson-stephens}), which is roughly half of the trivial upper bound, fails badly over finite fields, as it relies on the following lemma:

\begin{lemma}[\cite{atkinson-stephens}]
\label{simul-diag}
Given two arbitrary square matrices $A,B\in\C^{n\times n}$, there exist diagonal matrices $J,K\in\C^{n\times n}$ such that $A+J$ and $B+K$ are simultaneously diagonalizable.
\end{lemma}

Lemma \ref{simul-diag} crucially relies on $\C$ being infinite and algebraically closed. As a counterexample over a finite field $\F$, consider setting $B=\M{0&1\\&&\ddots\\&&&1\\&&&0}$;
then the geometric multiplicity of each eigenvalue of $B+K$ is at most one, because $\rk(B+K-\lambda I)\ge n-1$ for any $\lambda$ and $K$.
Since any square matrix over a finite field $\F$ trivially has $\le |\F|$ many eigenvalues, $B+K$ is not diagonalizable if $n>|F|$.

Additionally, the rank of a specific tensor can be change across different fields, and usually the rank seems to increase over smaller fields:

\begin{example}
The tensor $T=\M{\M{1\\&1}&\M{&1\\1}}$ (Ex. \ref{addition-mod-2-tensor}) satisfies
$\rk_\F\paren{T}=\cases{3 & \ch{\F}=2 \\ 2 & \textrm{else}}$.
\end{example}
\begin{proof}
The upper bound for the case $\ch{\F}=2$ can proven with the CPD 
$T=\M{1\\\phantom{0}}\otimes \M{1\\1}\otimes \M{1\\1}
+\M{1\\1}\otimes \M{1\\\phantom{0}}\otimes \M{\phantom{0}\\1}
+\M{1\\1}\otimes \M{\phantom{0}\\1}\otimes \M{1\\\phantom{0}}
=\M{\M{1&1\\1&1}&O}
+\M{\M{&1\\\phantom{0}}&\M{&1\\\phantom{0}}}
+\M{\M{&\phantom{0}\\1}&\M{&\phantom{0}\\1}}$.

The case $\ch{\F}\ne 2$ and a lower bound for the case $\ch{\F}=2$ can be seen by running Algorithm \ref{alg:main} with rank threshold $R=2$, which we demonstrate symbolically.

If $\rk(T)=2$, there must exist two linearly independent solutions $\M{a&b}$ to the equation $\rk(\M{a&b}\times_0 T)=1$. Since $\M{a&b}\times_0 T=\M{a&b\\b&a}$ is a $2\times 2$ matrix, the equation is equivalent to forcing the determinant of this matrix to be 0, so $a^2-b^2=0=(a+b)(a-b)$.
If $\ch{\F}\ne 2$, there exist solutions $\M{a&b}=\M{1&1},\ \M{1&-1}$. Otherwise, the polynomial equation degenerates to $(a+b)^2=0 \Rightarrow a+b=0$, so the set of solutions is contained in a 1-dimensional subspace.
\end{proof}

\begin{example}
\label{W-state-pow2}
The tensor $T=\M{
    \M{1\\&\phantom{0}\\&&\phantom{0}\\&&&\phantom{0}}
    &\M{&1\\1&\phantom{0}\\&&\phantom{0}\\&&&\phantom{0}}
    &\M{&&1\\&\phantom{0}\\1&&\phantom{0}\\&&&\phantom{0}}
    &\M{&&&1\\&&1\\&1\\1}
}$ satisfies
$\rk_\F\paren{T}=\cases{8 & \F=\F_2 \\ 7 & \ch{\F}\ne 2}$.
\end{example}
\begin{proof}
Over all fields, a rank lower bound of 7 can be proven by subbing out slices $(i,:,:)$ for each $0\le i<3$, then noticing that the resulting $1\times 4\times 4$ tensor is guaranteed to have rank 4.

For $\F=\F_2$, the exact rank of 8 is confirmed by running our implementation of Algorithm \ref{alg:main} with rref-pruning, lask-pruning, and rref-heuristic. A rank-8 CPD that we showed in \cite{yang25} is

\[
T=\cpdeval{
\M{
    1&0&0&0&0&0&0&0\\
    1&0&0&1&0&0&1&0\\
    1&0&1&0&0&1&0&0\\
    1&1&0&0&1&1&1&1\\
},\ 
\M{
    1&1&1&1&0&0&0&0\\
    0&0&0&1&0&0&1&1\\
    0&0&1&0&0&1&0&1\\
    0&1&0&0&1&0&0&0\\
},\ 
\M{
    1&1&1&1&0&0&0&0\\
    0&0&0&1&0&0&1&1\\
    0&0&1&0&0&1&0&1\\
    0&1&0&0&1&0&0&0\\
}
}.
\]

For $\ch{\F}\ne 2$, a rank-7 CPD of $T$ is shown in \cite{lysikov-slides}.
\end{proof}
\begin{remark}
We are not sure what the rank of the tensor in Example \ref{W-state-pow2} is over a field with characteristic 2 that is not $\F_2$.
\end{remark}
\begin{remark}
The tensor in Example \ref{W-state-pow2} is equal to $W\boxtimes W$, where $W=\M{\M{1\\&\phantom{0}}&\M{&1\\1}}$ is the ``W-state" tensor from Example \ref{W-state},
and $\boxtimes$ is the \textit{Kronecker product} of tensors, which is defined similarly as is for matrices.
\end{remark}

A potential problem with this phenomenon of tensor rank increasing over finite fields is that general lower bounds techniques, like the substitution method, may be too weak. Recall the indeterminate variables $\lambda_i$ in Theorem \ref{thm:sub-method}: to our knowledge, most mathematical proofs that use the substitution method try to avoid using these variables, or cancel them out when possible (e.g., the proof that $\rk(\ang{2,2,2})\ge 7$ due to \cite{brockett-dobkin} that we rephrase in Example \ref{mm222-rank7}). We call such proofs ``symbolic"; in contrast, stronger lower bounds may require enumerating some or all assignments of the indeterminates $\lambda_i$ and recursing, which we call ``combinatorial". Doing so may lead to a depth-first search problem with complexity comparable to our CPD search algorithm, with its own set of pruners that can speed up computation in practice.

\section{Conclusion}
We present algorithms for both ordinary CPD and border CPD search of an arbitrary tensor over a finite field, which run in exponential time much faster than na\"ive methods while only using polynomial space. For ordinary CPD, we also describe pruning conditions that can reduce the search space in practice for specific tensor inputs.

We also study the maximum possible rank of a tensor with a given shape over a specific field.
We rephrase existing results in a notation we believe is easier to read, and we improve some of these results by tuning parameters more carefully.
Finally, we use our aforementioned CPD search algorithm to exactly determine maximum rank of shapes $3\times 3\times 4$, $3\times 4\times 4$, and $3\times 3\times 5$ over the field $\F_2$.

We suspect that the runtime of our ordinary CPD search is near the limit for algorithms relying on linear analysis (e.g. contractions of the input tensor along a single axis, linear systems of linear equations) and that further improvements must use nonlinear methods.
Even so, the best known algorithms for this task (e.g., Gr\"obner bases) still use linear algebra at a high level by treating a polynomial as a linear combination of monomials, so a lot of challenging problems remain.

For border CPD search, there may be more room for improvement. The main difficulty we face is that a reduced row-echelon form of a matrix does not always exist over a border ring, but a more careful analysis of this ring may get around this obstacle.

\end{document}